\documentclass[journal,12pt,onecolumn,draftclsnofoot,]{IEEEtran}

\usepackage[latin1]{inputenc}
\usepackage{times,amsmath}
\usepackage{amsfonts}
\usepackage{pstool}
\usepackage{subfigure}
\usepackage{multirow}
\usepackage{enumerate}
\usepackage{graphicx}
\usepackage{MnSymbol}
\usepackage{stfloats}
\usepackage[table]{xcolor}
\usepackage[square, comma, sort&compress, numbers]{natbib}
\usepackage{nohyperref}
\usepackage{algorithm,algorithmic}
\usepackage{bigints}
\usepackage{amsmath}
\usepackage{amsthm}
\newtheorem{theorem}{Theorem}

\newtheorem{lemma}[theorem]{Lemma}
\newtheorem*{remark}{Remark}
\newtheorem{prop}{Proposition}
\newcounter{tempEquationCounter}
\newcounter{thisEquationNumber}

\DeclareMathOperator{\E}{\mathbb{E}}

\makeatletter
\newcommand{\vast}{\bBigg@{4}}
\newcommand{\Vast}{\bBigg@{5}}
\makeatother

\graphicspath{ {Figures/} }

\newcounter{relctr} 
\everydisplay\expandafter{\the\everydisplay\setcounter{relctr}{0}} 

\AtBeginDocument{} 

\begin{document}

\title{Effective Capacity Analysis of HARQ-enabled D2D Communication in Multi-Tier Cellular Networks}

\author{
\IEEEauthorblockN{Syed Waqas Haider Shah, \textit{Student Member, IEEE}, M. Mahboob Ur Rahman, \textit{Member, IEEE}, Adnan Noor Mian, \textit{Member, IEEE}, Octavia A. Dobre, \textit{Fellow, IEEE}, and Jon Crowcroft, \textit{Fellow, IEEE}}
}

\maketitle

\long\def\symbolfootnote[#1]#2{\begingroup%
\def\thefootnote{\fnsymbol{footnote}}\footnote[#1]{#2}\endgroup}
\symbolfootnote[0]{\hrulefill \\
Copyright (c) 2015 IEEE. Personal use of this material is permitted. However, permission to use this material for any other purposes must be obtained from the IEEE by sending a request to pubs-permissions@ieee.org.

Syed Waqas Haider Shah and Jon Crowcroft are with the Computer Lab, University of Cambridge, 15 JJ Thomson Avenue, Cambridge, UK CB3 0FD (\{sw920, jon.crowcroft\}@cl.cam.ac.uk). Syed Waqas Haider Shah is also with the Electrical Engineering Department, Information Technology University, Lahore 54000, Pakistan (waqas.haider@itu.edu.pk).

Muhammad Mahboob Ur Rahman and Adnan Noor Mian are with the Electrical Engineering Department, Information Technology University, Lahore 54000, Pakistan (\{mahboob.rahman, adnan.noor\}@itu.edu.pk

Octavia A. Dobre is with the Department of Electrical and Computer Engineering, Memorial University, St. John's, NL A1B 3X5, Canada (odobre@mun.ca)}

\maketitle

\begin{abstract}
This work does the statistical quality-of-service (QoS) analysis of a block-fading device-to-device (D2D) link in a multi-tier cellular network that consists of a macro-BS ($BS_{_{MC}}$) and a micro-BS ($BS_{_{mC}}$) which both operate in full-duplex (FD) mode. For the D2D link under consideration, we first formulate the mode selection problem---whereby D2D pair could either communicate directly, or, through the $BS_{_{mC}}$, or, through the $BS_{_{MC}}$---as a ternary hypothesis testing problem. Next, to compute the {\it effective capacity} (EC) for the given D2D link, we assume that the channel state information (CSI) is not available at the transmit D2D node, and hence, it transmits at a fixed rate $r$ with a fixed power. This allows us to model the D2D link as a Markov system with six-states. We consider both overlay and underlay modes for the D2D link. Moreover, to improve the throughput of the D2D link, we assume that the D2D pair utilizes two special automatic repeat request (ARQ) schemes, i.e., Hybrid-ARQ (HARQ) and truncated HARQ. Furthermore, we consider two distinct queue models at the transmit D2D node, based upon how it responds to the decoding failure at the receive D2D node.
Eventually, we provide closed-form expressions for the EC for both HARQ-enabled D2D link and truncated HARQ-enabled D2D link, under both queue models. Noting that the EC looks like a quasi-concave function of $r$, we further maximize the EC by searching for an optimal rate via the gradient-descent method. Simulation results provide us the following insights: i) EC decreases with an increase in the QoS exponent, ii) EC of the D2D link improves when HARQ is employed, iii) EC increases with an increase in the quality of self-interference cancellation techniques used at $BS_{_{mC}}$ and $BS_{_{MC}}$ in FD mode.
\end{abstract}

\IEEEpeerreviewmaketitle

\begin{IEEEkeywords}
Effective capacity, D2D communication, retransmission, automatic repeat request, hybrid-ARQ, quality-of-service.
\end{IEEEkeywords}

\section{Introduction}
In wireless communication, reliability is considered one of the key performance indicators for data transmission. It becomes more critical with the emergence of mission-critical and delay-sensitive communication paradigms and their potential applications in society, such as video streaming, online gaming, and augmented reality, etc. These communication paradigms strive to accommodate services with ultra-reliable and low latency requirements. The quality of the wireless channel, which is defined by shadowing, multi-path fading, and inter-user and inter-channel interference, affects the achievable reliability. Prior knowledge of the channel conditions at the transmitter plays an important role in achieving the required reliability. When the transmitter has the perfect channel state information (CSI) before the transmission, it adjusts its transmission power and the transmission rate according to the channel conditions. This way, the optimal performance of the channel can be achieved \cite{gross2012scheduling}. However, in practice, perfect knowledge of the CSI at the transmitter is hard to acquire due to rapidly changing wireless channel conditions (slow and fast fading). Therefore, in practical wireless systems, block fading channel models are used. In these models, pilot bits are transmitted at the start of each fading/time block to approximate the fading process of the channel. This fading process is supposed to remain the same for the entire fading/time block. However, if the block length is long or the pathloss changes rapidly, this approximation does not truly represent the entire fading/time block.

Device-to-device (D2D) communication, on the other hand, is a type of communication with opportunistic channel allocation \cite{tang2017ac}. In this type, a D2D device transmits data in either a direct-D2D mode, using overlay (orthogonal channel allocation) and underlay (non-orthogonal channel allocation) settings, or in a cellular mode (relaying through the base station) \cite{zhao2018caching}. The opportunistic nature of channel allocation in the D2D communication paradigm makes it hard (or not even feasible sometimes) to acquire CSI at a D2D transmitting device \cite{liu2018transceiver,liu2015outage}. Data transmission without prior knowledge of CSI at the transmitting device leads to an increase in the packet drop ratio due to rapidly changing channel conditions. To this end, multiple techniques can be used to ensure reliability, such as shortening the length of the time/fading block (to allow more retransmissions) or reducing the packet size. In particular, automatic repeat request (ARQ) and hybrid-ARQ (HARQ) schemes were proposed to enhance the reliability of the communication channel when CSI is not available at the transmitter prior to the transmission.

In the ARQ retransmission scheme, parity bits are added to the transmitted packets for error detection (ED) at the receiver. If the receiver detects an error, it sends a negative-acknowledgment (NACK) using an error-free feedback link; then, the transmitter retransmits the packet. Retransmission of the same packet continues until the transmitter receives a positive ACK. One of the major drawbacks of the ARQ scheme is that the throughput does not remain constant; instead, it falls rapidly as the channel conditions deteriorate (due to high retransmission frequency). To enhance the performance of the ARQ and to reduce the retransmission frequency, another scheme was introduced. This scheme, known as HARQ, uses forward error correction (FEC) codes, along with ARQ \cite{kim2008optimal}. HARQ is generally used in two settings, namely type-I HARQ and type-II HARQ. In the former, packets are encoded with ED and FEC codes before the transmission, and the receiver tries to remove the error using these codes when an erroneous packet is received (instead of sending NACK right away). Retransmission of the same packet is only requested when the receiver fails to decode the received packet. In the latter, the transmitter first sends data and ED codes only. If the receiver fails to decode the received packet, it sends a NACK, and then the transmitter sends ED and FEC codes in the second transmission attempt. If the packet still has an error, the receiver combines the information received in both transmissions for error correction \cite{burich2017cross}. This phenomenon is known as chase combining or soft combining. The transmitter keeps sending the same parity bits (ED and FEC codes) in each retransmission attempt. If the transmitter sends different parity bits every time it receives a NACK, it is known as type-III HARQ (also known as incremental redundancy) \cite{yafeng2003performance}. HARQ overall provides better performance in terms of throughput and reliability when compared to ARQ. In the case of D2D communication, HARQ can be used in both direct-D2D and cellular-D2D modes. In the direct-D2D mode, a D2D receiver sends ACK/NACK directly to the transmitter in either overlay or underlay settings depending upon the allocated channel. In the cellular-mode, the D2D receiver first sends ACK/NACK to the base station (BS), which the BS then relays to the D2D transmitter. The cellular mode allows HARQ to reuse the existing downlink and uplink channels with minimal changes, at the cost of additional overhead and possibly a longer delay in feedback. Throughput analysis is one of the most common tools used to measure the performance of the retransmission schemes. However, for D2D communication or other delay-sensitive wireless applications, throughput analysis may not provide the required delay guarantees. Moreover, throughput varies with varying channel conditions and drops quickly when the channel conditions deteriorate.

In delay-sensitive wireless applications, it is desirable to have system throughput subject to given quality-of-service (QoS) requirements. The Effective Capacity (EC) is an analytical tool to find the maximum constant arrival rate that can be supported by the time-varying channel conditions while satisfying the statistical QoS guarantees imposed at the transmitter's queue \cite{wu2003effective}. It provides statistical QoS guarantees for throughput in terms of delay bounds. The EC has been used for various wireless channels, including cognitive radios \cite{musavian2010effective}, two-hop wireless channels \cite{qiao2012effective}, D2D \cite{shah2019impact}, licensed-unlicensed interoperable D2D \cite{shah2020statistical}, MIMO wireless networks \cite{cheng2013qos}, and underwater acoustic channels \cite{aman2020effective}. More recently, the EC analysis has also been performed for different retransmission schemes \cite{larsson2016effective,hu2020throughput,li2016throughput}. However, to the best of the authors' knowledge, this is the first study which provides the EC analysis of HARQ-enabled D2D communication in multi-tier future cellular networks.

More specifically, this work provides the following contributions:
\begin{itemize}
  \item We formulate a mode selection mechanism for D2D communication in multi-tier cellular networks as a ternary hypothesis testing problem and compute the corresponding error and correct-detection probabilities (Section III). This mechanism selects a communication mode among the three available modes (direct-D2D, micro-cell D2D, and macro-cell D2D) based on the pathloss measurements of the transmission link.
  \item We perform the EC analysis of HARQ-enabled D2D communication in multi-tier cellular networks. We also provide an analysis of the impact of the mode selection mechanism on the EC of HARQ-enabled multi-tier D2D communication. We assume that the CSI is not available at the transmit D2D node, and hence, it transmits at a fixed rate with a fixed power. It allows us to model the D2D link as a Markov system with six-states. We then perform the Markov chain modeling of the D2D link in both overlay and underlay settings.
  \item We provide the EC analysis of HARQ-enabled D2D communication for two distinct queue models at the transmit D2D node, based upon how it responds to the decoding failure at the receive D2D node. We provide closed-form expressions for the EC of HARQ-enabled D2D link under both queue models.
  \item We propose a special-case of truncated HARQ-enabled D2D communication in which a transmitting device transmits a packet only twice. It transmits in underlay settings in the first attempt, and if the receiver fails to decode the received packet successfully, it retransmits the same packet in overlay settings. If the receiver fails to decode the packet in the second transmission attempt, the transmitting device either drops the packet or lowers the transmission priority of the packet (based on the queue model in use). We then perform the EC analysis and provide the closed-form expressions for the EC of truncated HARQ-enabled D2D communication under both queue models.
  \item Lastly, we provide closed-form expressions for the optimal transmission rates for our proposed case of truncated-HARQ enabled D2D communication under both queue models.
\end{itemize}

The remainder of this paper is organized as follows. Section II presents the system model for our proposed multi-tier D2D communication and some background knowledge of EC, full-duplex, and ARQ. Section III introduces the mode selection mechanism for the proposed model. Section IV provides the EC analysis. Sections IV-A and IV-B describe the Markov chain modelling for the proposed ternary hypothesis testing (THT) problem. Section IV-C and IV-D present the EC of HARQ-enabled multi-tier D2D and of the truncated HARQ case of multi-tier D2D, respectively. Section V provides a detailed numerical investigation using simulation results. Finally, the paper concludes in Section VI.

\section{System Model and Background}
\subsection{System Model}
We consider a two-tier cellular network scenario in which a micro-cell (mC) BS is deployed in a coverage region of a macro-cell (MC) BS, as shown in Fig. \ref{system_model}. In a 5G multi-tier network architecture, MC-BS and mC-BS usually operate on lower frequencies and higher millimeter-wave frequencies, respectively \cite{panwar2016survey}. Therefore, they do not experience inter-tier interference.\footnote{In scenarios where all the tiers in a multi-tier network architecture use the same frequency spectrum, one needs to consider inter-tier interference for calculating the respective channel capacities \cite{6845056}.} MC-BS provides low-rate connectivity to a large number of users in a wide coverage area. On the other hand, mC-BS provides high data rate connectivity to a small number of users in a limited coverage area. In two-tier cellular networks, a D2D transmitting device can communicate with its receiver in three possible communication modes. It can either communicate directly (direct-D2D mode) or by relaying its data through MC-BS (MC-D2D mode) or mC-BS (mC-D2D mode), as shown in Fig. \ref{system_model}. It can also use either underlay (reusing the cellular user's resources) or overlay (using orthogonal resource blocks) settings for data transmission based on the network conditions. This problem of selecting a communication mode from the available ones is known as mode selection \cite{shah2019impact}.
\begin{figure}[ht]
\begin{center}
	\includegraphics[width=3.2in]{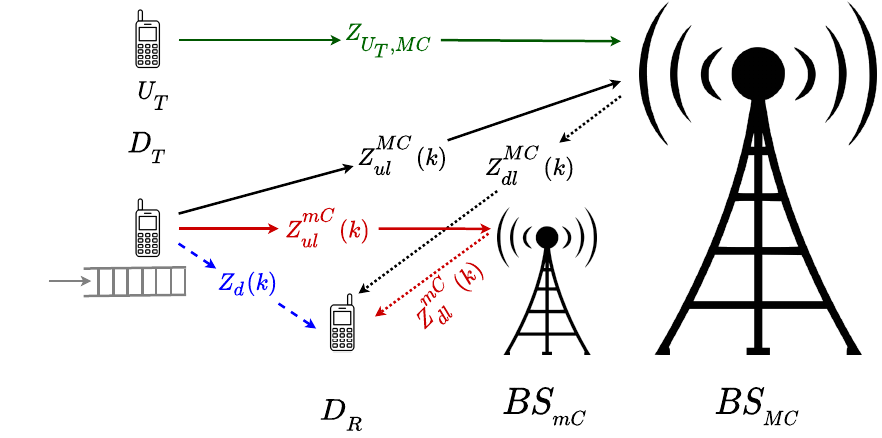}
\caption{System model: D2D communication in multi-tier cellular networks. $D_T$ communicates with $D_R$ in direct-D2D mode (shown as blue dotted arrows), mC-D2D mode (shown as red arrows), or MC-D2D mode (shown as black arrows); solid and dotted arrows show the uplink and downlink transmissions, respectively.}
\label{system_model}
\end{center}
\end{figure}

We also make the following assumptions for our analysis: i) direct-D2D, mC-D2D, and MC-D2D channels are block-fading channels that have Rayleigh distribution, and fading remains constant for each block, changing independently between blocks; ii) both the mC-BS and MC-BS use decode-and-forward operation to relay data to $D_R$ in mC-D2D and MC-D2D communication modes; iii) both mC-BS and MC-BS operate in full-duplex mode \cite{al2017full}, so we use the residual self-interference (SI) as a factor of noise in our analysis.

\subsection{Background}
\paragraph*{\textbf{Effective Capacity (EC)}}
EC is the maximum constant arrival rate that can be supported by the time-varying channel while satisfying the statistical QoS guarantees imposed as delay constraints at the transmitter's queue. It is defined as the log moment generating function (MGF) of the cumulative channel service process \cite{wu2003effective}:
\begin{equation}
    EC = -\frac{\Lambda (-\theta)}{\theta} = -\lim_{t\to\infty} \frac{1}{\theta t}\log \E[e^{-\theta \sum_{k=1}^{t}s(k)}]
\end{equation}
where $s(k)$ is the channel service process in slot $k$, $\E[.]$ is the expectation operator, and $\theta$ is the QoS exponent. $\theta \to \infty$ ($\theta \to 0$) refers to delay-sensitive (delay-tolerant) communication.

\paragraph*{\textbf{Full-Duplex Communication}}
In full-duplex communication, nodes can transmit and receive at the same frequency and at the same time, therefore theoretically, the communication link can achieve double throughput. In a full-duplex system, the transmit signal interferes with the receive signal, thus introduce a SI. Generally, the SI cancellation is performed in two stages. In the first stage, passive cancellation techniques, such as antenna-separation and antenna-shielding are used \cite{everett2014passive}. In the second stage, active cancellation techniques, which can be digital or analog, are used \cite{elsayed2020low,le2020beam,ahmed2015all}. However, a complete SI cancellation is impossible in practical full-duplex systems \cite{jain2011practical}. Therefore, a residual SI can still be experienced at the transmit node even after employing these cancellation techniques. To this end, we use the residual SI in our analysis as a factor of noise.

\paragraph*{\textbf{Automatic Repeat Request (ARQ)}}
In ARQ, parity bits are added to the transmitted packets for error detection at the receiver node. If the receiver node detects an error, it sends a NACK, and the transmitter retransmits the packet. There are also some variants of ARQ, such as go-back-N, stop-and-wait, and selective repeat. In HARQ, FEC codes are also added along with parity bits to the transmitted packet \cite{zhao2005practical}. In this protocol, the receiver node first tries to remove the error using the FEC codes when an erroneous packet is received, rather than sending NACK right away. Retransmission of the packet continues until the receiver node successfully decodes the received packet. HARQ is generally used in three different settings, explained in Section I. Additionally, if an upper limit is set for the packet's retransmission attempts, it is called truncated HARQ \cite{malkamaki2000performance}. Network coding can also be used to enhance the performance of HARQ in wireless broadcasting and multi-user networks, such as network-coded HARQ (NC-HARQ) \cite{ahmad2018analysis} and network-turbo-coding based HARQ \cite{xu2014ntc}. Basic NC-HARQ protocols may increase the computational complexity and delay. It can be avoided using low-complexity turbo coding techniques \cite{chen2013survey}. Moreover, to enhance the throughput of NC-HARQ even further, adaptive random network coding (ARNC) can be used \cite{hu2017arnc}. It adaptively encodes multiple packets with the highest priority in each time slot.

\section{Mode Selection}

The problem of mode selection at the transmit device $D_T$ is basically choosing the best transmission path among a set of candidate paths. For the considered system model, mode selection implies selection between direct path ($D_T$ $\to$ $D_R$), via micro-BS ($D_T \to BS_{mC} \to D_R$), and via macro-BS ($D_T \to BS_{MC} \to D_R$). Mode selection is traditionally feature-based whereby the features of the candidate channels (e.g., received signal strength, instant CSI, statistical CSI, instant signal to noise ratio, etc.) are utilized to select the most suitable channel for transmission during upcoming uplink slot. Furthermore, since the acquisition of instant CSI is quite demanding, this work does mode selection based upon statistical CSI (i.e., pathloss) only. \footnote{Statistical CSI (pathloss) is used as the sole feature for mode selection because it varies slowly in the wireless channel, and once estimated, can last for multiple seconds. On the other hand, instantaneous CSI changes quickly due to small-scale fading (if the wireless channel is stationary even then, small-scale fading needs to be estimated multiple times in one second). Moreover, the overhead associated with the channel estimation for instantaneous CSI is also large due to the channel training.} In our system model, $BS_{_{MC}}$ performs the mode selection mechanism. Specifically, during time slot $k$, the pathloss for all the three candidate channels ($D_T \to D_R$, $D_T \to BS_{_{mC}}$, and $D_T \to BS_{_{MC}}$) is measured by $D_R$, $BS_{_{mC}}$, and $BS_{_{MC}}$, respectively (see Appendix \ref{pathlossprop}). All the three pathloss measurements reach $BS_{_{MC}}$, which performs mode selection for the upcoming time slot ($k+1$ time slot). In short, $BS_{_{MC}}$ does the mode selection for time slot $k+1$ based upon the pathloss measurements of the current time slot (time slot $k$). Thus, by mode selection, $BS_{_{MC}}$ chooses the communication link with the smallest estimated pathloss and then conveys this information to $D_T$ through a downlink control channel. Further, because the proposed mode selection problem selects a communication mode based on the estimated pathloss measurements, we provide a step-by-step procedure for pathloss estimation in Appendix \ref{pathlossprop}.

\subsection{Ternary Hypothesis Testing (THT)}
The mode selection problem is formulated as the following THT problem.
\begin{equation}
	\label{eq:H0H1}
	 \begin{cases} H_0: & \text{direct-D2D mode ($D_T \to D_R$)} \\
                   H_1: & \text{micro cell (mC)-D2D mode ($D_T \to BS_{_{mC}} \to D_R$)} \\
                   H_2: & \text{macro-cell (MC)-D2D mode ($D_T \to BS_{_{MC}} \to D_R$).} \end{cases}
\end{equation}
Where the hypothesis $H_0$, $H_1$, $H_2$ states that communication via direct link, via micro-BS, via macro-BS is most suitable for transmission during the upcoming slot.

Let $L_d$, $L_{mC}$, $L_{MC}$ represent the true pathloss of $D_T \to D_R$, $D_T \to BS_{mC}$, and $D_T \to BS_{MC}$ links, respectively. Moreover, let $\widehat{L}_d$, $\widehat{L}_{mC}$, $\widehat{L}_{MC}$ represent the noisy measurement of $L_d$, $L_{mC}$, $L_{MC}$. Appendix \ref{pathlossprop} provides a step-by-step procedure for calculating the noisy measurement of pathloss for all the three candidate links. According to the mode selection problem, the direct-D2D mode will be selected when the estimated pathloss of $D_T \to D_R$ ($\widehat{L}_d$) link is the smallest among the estimated pathlosses of the candidate links. Similarly, mC-D2D and MC-D2D modes will be selected when the estimated pathloss of $D_T \to BS_{mC}$ ($\widehat{L}_{mC}$), and $D_T \to BS_{MC}$ ($\widehat{L}_{MC}$) link is the smallest, respectively. Now, the THT problem in \eqref{eq:H0H1} could be re-cast as follows:
\begin{equation}
	\label{eq:H0H1_2}
	 \begin{cases} H_0: & \widehat{L}_d=\min\big\{\widehat{L}_d, \widehat{L}_{_{mC}}, \widehat{L}_{_{MC}}\big\} \\
                   H_1: & \widehat{L}_{_{mC}}=\min\big\{\widehat{L}_d, \widehat{L}_{_{mC}}, \widehat{L}_{_{MC}}\big\} \\
                   H_2: & \widehat{L}_{_{MC}}=\min\big\{\widehat{L}_d, \widehat{L}_{_{mC}}, \widehat{L}_{_{MC}}\big\}, \end{cases}
\end{equation}
where $\widehat{L}_d\sim \mathcal{N} (L_d,\sigma^2)$, $\widehat{L}_{_{mC}}\sim \mathcal{N} (L_{_{mC}},\sigma^2)$, and $\widehat{L}_{_{MC}}\sim \mathcal{N} (L_{_{MC}},\sigma^2)$ are the probability distribution of the noisy measurement of pathloss in direct-D2D, mC-D2D, and MC-D2D modes, respectively (see Appendix A). From eq. \eqref{eq:H0H1_2}, we can see that $H_0$ will be selected when the noisy measurement of the pathloss of $D_T \to D_R$ link ($\widehat{L}_d$) is the smallest. Similarly $H_1$ and $H_2$ will be selected when $\widehat{L}_{_{mC}}$ and $\widehat{L}_{_{MC}}$ are the smallest among the candidate links' pathlosses, respectively.

Let $\mathbf{l}=[L_d,L_{mC},L_{MC}]^T$. Also, let $\mathbf{l}^{(s)}=\text{sort}(\mathbf{l})$, where sort(.) operator sorts the elements of a vector in ascending order. Let $\mathbf{l}^{(s)}=[L_A,L_B,L_C]^T$; thus, $L_A < L_B < L_C$ (see Fig. \ref{THT}). In other words, $\mathbf{l}$, $\mathbf{l}^{(s)}$ are $3\times1$ vector each that contain the unsorted pathlosses, and sorted pathlosses of the three candidate links, respectively. Then, the following holds: $\hat{L}_A\sim \mathcal{N}(L_A,\sigma^2)$, $\hat{L}_B\sim \mathcal{N}(L_B,\sigma^2)$, $\hat{L}_C\sim \mathcal{N}(L_C,\sigma^2)$, where $\hat{L}_A$, $\hat{L}_B$, and $\hat{L}_C$ denote the noisy measurements of $L_A$, $L_B$, and $L_C$, respectively. Then, the THT problem for the sorted pathlosses could be formulated as the following two log-likelihood ratio tests (LLRT) (see section 3.2 of \cite{madhow2008fundamentals}):
\begin{subequations}\label{LLRT}
\begin{align}
  &\log_e (f_{\widehat{L}_A}(\widehat{l}_A) \underset{H_A}{\overset{H_B}{\gtrless}} \log_e (f_{\widehat{L}_B}(\widehat{l}_B))\\
  &\log_e (f_{\widehat{L}_B}(\widehat{l}_B)) \underset{H_B}{\overset{H_C}{\gtrless}} \log_e (f_{\widehat{L}_C}(\widehat{l}_C)),
  \end{align}
\end{subequations}
where $f_X(x)$ represents the probability density function (pdf) of the random variable $X$. (\ref{LLRT}a) states that when the pdf of the estimated pathloss $\bar{L}_A$ is smaller than the pdf of the estimated pathloss $\bar{L}_B$, $H_A$ will be selected, and vice-versa. Similarly, (\ref{LLRT}b) represents that when the pdf of the estimated pathloss $\bar{L}_B$ is smaller than the pdf of the estimated pathloss $\bar{L}_C$, $H_B$ will be selected, and vice-versa.

\begin{figure}[ht]
\begin{center}
	\includegraphics[width=3in]{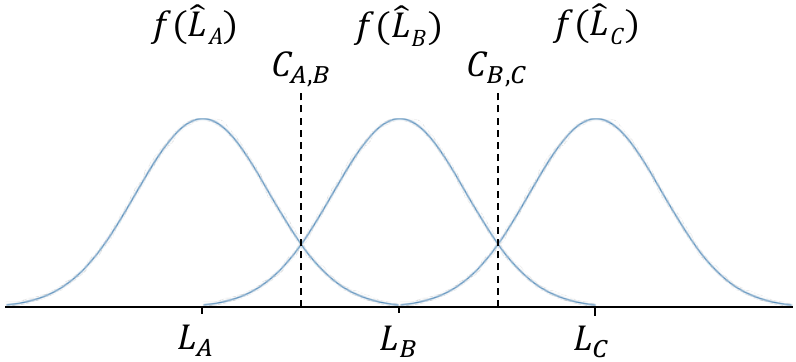}
\caption{The pdfs $f(\widehat{L}_A)$, $f(\widehat{L}_B)$, and $f(\widehat{L}_C)$: $C_{A,B}$ and $C_{B,C}$ are the decision thresholds; $L_A$, $L_B$, and $L_C$ are the true (but ordered) pathloss values of the three candidate links. }
\label{THT}
\end{center}
\end{figure}
\subsection{Performance of THT}
 We evaluate the performance of THT by using the correct-detection and error probabilities. Let $C_{A,B}$ and $C_{B,C}$ represent the decision thresholds (see Fig. \ref{THT}).Then, the three probabilities of correct-detection are given as:
\begin{subequations}\label{Pd_123}
\begin{align}
\begin{split}
P_{_{d,A}} & = \mathbb{P}(\widehat{L}_A<C_{A,B})\\
&= 1- Q\big(\frac{C_{A,B} - L_A}{\sigma}\big)
\end{split}\\
\begin{split}
P_{_{d,B}}  & = \mathbb{P}(C_{A,B}<\widehat{L}_B<C_{B,C} )\\
&= Q\big(\frac{C_{A,B}-L_B}{\sigma}\big)-Q\big(\frac{C_{B,C}-L_B}{\sigma}\big)
\end{split}\\
\begin{split}
P_{_{d,C}}  & = \mathbb{P}(\widehat{L}_C>C_{B,C})\\
&= Q\big(\frac{C_{B,C} - L_C}{\sigma}\big),
\end{split}
\end{align}
\end{subequations}
where $Q(x)=\frac{1}{\sqrt{2\pi}} \int_x^\infty  e^{-\frac{t^2}{2}} dt$ is the complementary cumulative distribution function (CCDF) of a standard normal random variable. (\ref{Pd_123}a), (\ref{Pd_123}b), and (\ref{Pd_123}c) represent the correct detection probabilities of selecting $H_A$, $H_B$, and $H_C$, respectively. More specifically, $P_{d,A}$ corresponds to the probability of the scenario when the estimated value of the smallest pathloss from the sorted pathloss vector ($1^{(s)}$) is smaller than $C_{A,B}$ (threshold between the pdfs of $\bar{L}_A$ and $\bar{L}_B$). In other words, this shows the probability that the mode selection mechanism selects $H_A$ when $\bar{L}_A$ was the smallest. Similarly, $P_{d,B}$ corresponds to the probability that the mode selection mechanism selects $H_B$ when $\bar{L}_B$ is smaller than $L_C$ and greater than $L_A$. Lastly, $P_{d,C}$ represents the probability that the mode selection mechanism selects $H_C$ when $\bar{L}_C$ was the biggest estimated pathloss among the estimated pathloss values of the three candidate links.  Moreover, the THT mechanism also incurs three kinds of errors, i.e., $P_{_{e,A}}=1-P_{_{d,A}}$, $P_{_{e,B}}=1-P_{_{d,B}}$, and $P_{_{e,C}}=1-P_{_{d,C}}$.

So far, we have computed the error and correct-detection probabilities for the ordered/sorted pathloss values ($L_A$, $L_B$, and $L_C$). However, the actual hypothesises are based on unsorted pathloss values. To this end, a relation needs to be established among the error and correct-detection probabilities of sorted/ordered pathloss values with the error and correct-detection probabilities of unsorted pathloss values. Let $P_{_{d,H_0}}$ ($P_{_{e,H_0}}$) represents the correct-detection (error) probability for selecting the direct-D2D mode. $P_{_{d,H_0}}$ shows that the direct-D2D link was the best (pathloss of the direct-D2D link was the smallest among all three pathloss), and the mode selection problem also detects the direct-D2D link. Whereas, $P_{_{e,H_0}}$ shows that the direct-D2D link was the best, but the mode selection problem makes an error and selects either mC-D2D or MC-D2D links for packet transmission. Similarly, $P_{_{d,H_1}}$ ($P_{_{e,H_1}}$) and $P_{_{d,H_2}}$ ($P_{_{e,H_2}}$) represent the correct-detection (error) probabilities for selecting mC-D2D and MC-D2D modes, respectively. \footnote{Ideally, the mode selection mechanism should be based on the true pathloss values of the candidate communication links. However, according to fundamental principles of statistical inference, the true pathloss can never be measured (since the received signal itself is corrupted with additive white Gaussian noise and channel fading). Therefore, the mode selection mechanism is based on the estimated pathloss values of the three communication links. These pathloss measurements come with some uncertainty (Gaussian, to be specific, shown in Appendix \ref{pathlossprop}), and due to this, the mode selection will not always be error-free. In other words, the uncertainty in the pathloss measurements introduces the error. Thus, the errors can never be made zero, but the hypothesis testing mechanism computes the thresholds in a way that these errors are minimized.} Further, to understand the relation between the error and correct-detection probabilities given in \eqref{Pd_123} and the probabilities for the error and the correct-detection of the actual hypothesises ($H_0$, $H_1$ and $H_2$), we provide the following example.
\paragraph*{\textbf{Example}}{Let $L_d = 90.7$, $L_{_{mC}} = 80.9$, and $L_{_{MC}} = 85.4$. Thus, $\mathbf{l} = [90.7, 80.9, 85.4]^T$. Then, $\mathbf{l}^{(s)} = \text{sort}(l) = [80.9, 85.4, 90.7]^T$. Thus, $L_A = 80.9$, $L_B = 85.4$, and $L_C = 90.7$. Furthermore, let $\sigma = 1$. Then, $\widehat{L}_d \sim \mathcal{N}(L_d, \sigma^2)$. Thus, $\widehat{L}_d \sim \mathcal{N}(90.7, 1)$. Similarly, $\widehat{L}_{_{mC}} \sim \mathcal{N}(80.9, 1)$ and $\widehat{L}_{_{MC}} \sim \mathcal{N}(85.4, 1)$. For measurements $\widehat{L}_A$, $\widehat{L}_B$, and $\widehat{L}_C$ of sorted pathloss values, we could write: $\widehat{L}_A \sim \mathcal{N}(L_A, \sigma^2)$. Thus, $\widehat{L}_A \sim \mathcal{N}(80.9, 1)$. Similarly, $\widehat{L}_B \sim \mathcal{N}(85.4, 1)$ and $\widehat{L}_C \sim \mathcal{N}(90.7, 1)$. Then, the correct-detection probabilities are $P_{_{d,A}} = 1-Q(2.25) = 0.988$, $P_{_{d,B}} = Q(-2.25)-Q(2.5) = 0.981$, and $P_{_{d,C}} = Q(-2.8) = 0.997$. Next, recall the following mapping due to the sort operation: $L_A = L_{_{mC}}$, $L_B = L_{_{MC}}$, and $L_C = L_{d}$. Thus, $P_{_{d,H_1}} = P_{_{d,A}} = 0.988$, $P_{_{d,H_2}} = P_{_{d,B}} = 0.981$, and $P_{_{d,H_0}} = P_{_{d,C}} = 0.997$. Similarly, $P_{_{e,H_1}} = P_{_{e,A}} = 0.012$, $P_{_{e,H_2}} = P_{_{e,B}} = 0.019$, and $P_{_{e,H_0}} = P_{_{e,C}} = 0.003$.} $\blacksquare$

Using the error and correct-detection probabilities of hypothesis $H_0$, $H_1$, and $H_2$, one can measure the performance of the mode selection mechanism. Next, we perform the statistical QoS analysis for HARQ-enabled D2D communication and observe the impact of mode selection on the analysis.
\section{Effective Capacity Analysis}
In our analysis, we consider $D_T$ is unaware of CSI prior to the transmission; therefore, it transmits using a fixed transmit power $\bar{P}$ at a fixed rate $r$ (bits/sec). Consequently, for each of the three hypotheses (direct-D2D, mC-D2D, and MC-D2D modes), the D2D link is considered ON when the instantaneous channel capacity of the link is greater than the fixed transmission rate of $D_T$; otherwise, the D2D link is considered in the OFF condition. To sum things up, due to the mode selection and the nonavailability of CSI at the transmitter (CSIT), one can model the D2D link as a Markovian process. Below, we describe the details of the Markov chain modelling of the D2D link for the overlay scenario and the underlay scenario, followed by the EC analysis of HARQ-enabled D2D communication.
\subsection{Markov Chain Modelling of Overlay-D2D}
Let us consider $C_d(k)$, $C_{_{mC}}(k)$, and $C_{_{MC}}(k)$ as the instantaneous channel capacities, during time slot $k$, of the direct-D2D, mC-D2D, and MC-D2D links, respectively. When $r<C_d(k)$, $r<C_{_{mC}}(k)$, and $r<C_{_{MC}}(k)$, the direct-D2D, mC-D2D, and MC-D2D links, respectively, transmit $r$ bits/sec; thus, they are considered as being in the ON state. On the other hand, when  $r>C_d(k)$, $r>C_{_{mC}}(k)$, and $r>C_{_{MC}}(k)$, the direct-D2D, mC-D2D, and MC-D2D links, respectively, transmit $0$ bits/sec; thus, they are considered as being in the OFF state. This leads to the six-state Markovian process, as shown in Table \ref{states},
\begin{table}[b]
\centering
\caption{Markov Chain Representation of Six States.}
\label{states}
\resizebox{\columnwidth}{!}{\begin{tabular}{|
>{\columncolor[HTML]{EFEFEF}}c |c|c|c|}
\hline \hline
State & \cellcolor[HTML]{EFEFEF}Description                                                         & \cellcolor[HTML]{EFEFEF}Notation                                           & \cellcolor[HTML]{EFEFEF}Action                                                                          \\ \hline
$s_1$  & \begin{tabular}[c]{@{}c@{}}Direct-D2D mode is\\ selected and the link is ON\end{tabular}    & \begin{tabular}[c]{@{}c@{}}$H_0$ \&\\ $r < C_d(k)$\end{tabular} & \begin{tabular}[c]{@{}c@{}}decoding successful at $D_R$,\\ $r$ bits received\end{tabular} \\ \hline
$s_2$  & \begin{tabular}[c]{@{}c@{}}Direct-D2D mode is\\ selected and the link is OFF\end{tabular}    & \begin{tabular}[c]{@{}c@{}}$H_0$ \&\\ $r > C_d(k)$\end{tabular} & \begin{tabular}[c]{@{}c@{}}decoding failure at $D_R$,\\ $0$ bits received\end{tabular}                                                                                        \\ \hline
$s_3$  & \begin{tabular}[c]{@{}c@{}}mC-D2D mode is\\ selected and the link is ON\end{tabular} & \begin{tabular}[c]{@{}c@{}}$H_1$ \&\\ $r < C_{_{mC}}(k)$\end{tabular} & \begin{tabular}[c]{@{}c@{}}decoding successful at $D_R$,\\ $r$ bits received\end{tabular} \\ \hline
$s_4$  & \begin{tabular}[c]{@{}c@{}}mC-D2D mode is\\ selected and the link is OFF\end{tabular}     & \begin{tabular}[c]{@{}c@{}}$H_1$ \&\\ $r > C_{_{mC}}(k)$\end{tabular} & \begin{tabular}[c]{@{}c@{}}decoding failure at $D_R$,\\ $0$ bits received\end{tabular} \\ \hline
$s_5$  & \begin{tabular}[c]{@{}c@{}}MC-D2D mode is\\ selected and the link is ON\end{tabular}    & \begin{tabular}[c]{@{}c@{}}$H_2$ \&\\ $r < C_{_{MC}}(k)$\end{tabular} & \begin{tabular}[c]{@{}c@{}}decoding successful at $D_R$,\\ $r$ bits received\end{tabular} \\ \hline
$s_6$  & \begin{tabular}[c]{@{}c@{}}MC-D2D mode is\\ selected and the link is OFF\end{tabular} & \begin{tabular}[c]{@{}c@{}}$H_2$ \&\\ $r > C_{_{MC}}(k)$\end{tabular} & \begin{tabular}[c]{@{}c@{}}decoding failure at $D_R$,\\ $0$ bits received\end{tabular}                                                                                         \\ \hline
\end{tabular}}
\end{table}
The instantaneous channel capacity of the direct-D2D link is,
\begin{equation}
 \label{eq:Cd}
{C^{o}_{d}}(k) = B \log_2 \bigg(1+\frac{\bar{P} Z_d(k)}{L_{d}(k) N_0}\bigg) = B \log_2 \big(1+\gamma_d(k)\big)
\end{equation}
where $Z_d(k)$ and $L_{d}(k)$ represent the channel coefficients and the pathloss of the direct-D2D link in time slot $k$, respectively, $B$ represents the bandwidth allocated to the transmit D2D node, and $\gamma_d(k)$ represent the signal-to-noise-ratio (SNR) of the direct-D2D link in time slot $k$. Before finding the instantaneous channel capacity for the mC-D2D link, we note that $BS_{_{mC}}$ operates in full-duplex mode. Therefore, to find it's channel capacity, we have the following proposition \ref{prop_outage_d}.
\begin{prop}\label{prop_outage_d}
  The instantaneous channel capacity of full-duplex enabled mC-D2D link ($D_T \to BS_{mC} \to D_R$) in overlay settings is,
  \begin{equation*}
    {C^{o}_{_{mC}}}(k)= B \log_2 \big(1+\gamma_{_{mC}}(k)\big)
  \end{equation*}
  where $\gamma_{_{mC}}(k) = \min \big\{ \gamma^{^{mC}}_{_{ul}}(k), \gamma^{^{mC}}_{_{dl}}(k) \big\}$ is the net-SNR of the mC-D2D link, and $\gamma^{^{mC}}_{_{ul}}(k)$ and $\gamma^{^{mC}}_{_{dl}}(k)$ are the SNRs of the uplink and the downlink channels of mC-D2D mode, respectively.
\end{prop}
\begin{proof}
  Given in Appendix \ref{prop1}.
\end{proof}

Similarly, one can find the instantaneous channel capacity for full-duplex enabled MC-D2D link ($D_T \to BS_{MC} \to D_R$) in overlay settings ($C^{o}_{_{MC}}(k)$) by following the similar steps given in Appendix \ref{prop1}. Consequently, it turns out to be,
\begin{equation}\label{eq:CMc}
C^{o}_{_{MC}}(k) = B \log_2 \big(1+\gamma_{_{MC}}(k)\big).
\end{equation}
Where $\gamma_{_{MC}}(k) = \min \big\{ \gamma^{^{MC}}_{_{ul}}(k), \gamma^{^{MC}}_{_{dl}}(k) \big\}$ is the net-SNR of the MC-D2D link, and $\gamma^{^{MC}}_{_{ul}}(k)$ and $\gamma^{^{MC}}_{_{dl}}(k)$ are the SNRs of the uplink and the downlink channels of MC-D2D mode, respectively.

Next, we find the state transition probabilities for states $s_1,s_2,s_3,s_4,s_5$, and $s_6$, as shown in Table \ref{states}. Let $p_{i,j}=[\mathbf{P}_o]_{i,j}$ be the transition probability from state $i$ to state $j$, with $\mathbf{P}_o$ as the transition probability matrix for overlay-D2D. Due to the block-fading nature of the channel, state change for the D2D link occurs in every timeblock. Now, we calculate the state transition probabilities for the Markov chain model, starting with the following: \footnote{Note that, state transition probability for each state depends upon two factors; the decision of the mode selection problem and the condition on the transmission rate.}
\begin{equation}
p_{1,1} = \mathbb{P}\big\{  H_0(k) \; \& \; r<C_{d_o}(k)  \big|  H_0(k-1) \; \& \; r<C_{d_o}(k-1)  \big\}.
\end{equation}
The condition on the transmission rate can also be translated into the SNR of the transmission link lower bounded by a minimum required value of SNR. This is shown in the following:
\begin{equation}
p_{1,1} = \mathbb{P}\big\{  H_0(k) \; \& \; \gamma_d(k)>\gamma_{req}  \big|  H_0(k-1) \; \& \; \gamma_d(k-1)>\gamma_{req}  \big\},
\end{equation}
where $\gamma_{req}=2^{r/B}-1$. Because the mode selection process is independent of the fading process $\{\gamma_d\}_k$, we can write:
\begin{equation}
p_{1,1} = \mathbb{P}\big\{  H_0(k)\big|H_0(k-1) \big\} \mathbb{P}\big\{ \gamma_d(k)>\gamma_{req}\big|\gamma_d(k-1)>\gamma_{req}  \big\}.
\end{equation}
Moreover, we note that the fading process $\{\gamma_d\}_k$ as well as the mode selection process are memoryless (because these processes change independently between time slots). Specifically, $\mathbb{P}(H_0(k)|H_{y}(k-1))=\mathbb{P}(H_0(k))$ for $y \in \{0,1,2\}$, and $\mathbb{P}(\gamma_d(k)|\gamma_d(k-1))=\mathbb{P}(\gamma_d(k))$. Therefore,
\begin{equation}
\label{eq:p11}
p_{1,1} = \mathbb{P}\big(H_0(k)\big) \mathbb{P}\big(\gamma_d(k)>\gamma_{req}\big),
\end{equation}
where $\mathbb{P}(H_0(k))=\mathbb{P}(H_0|H_0)+\mathbb{P}(H_0|H_1)+\mathbb{P}(H_0|H_2)$, and $\mathbb{P}(H_0|H_0) = P_{_{d,H_0}}$. Because the SNR $\gamma_d(k)$ is exponentially distributed, $\mathbb{P}(\gamma_d(k)>\gamma_{req}) = 1 - \mathbb{P}(\gamma_d(k)<\gamma_{req}) = e^{-\gamma_{req}/\E(\gamma_d(k))}$, where $\E(\gamma_d(k))=\frac{\bar{P}}{L_{d} N_0}$. Now, one can see that the transition probability $p_{1,1}$ does not depend on the original state. Therefore, $p_{i,1} = p_1$.
Similarly,
\begin{equation}
\label{eq:p2p3p4}
\begin{split}
p_{i,2} &= p_2 = \mathbb{P}\big(H_0(k)\big) \mathbb{P}\big(\gamma_d(k)<\gamma_{req}\big) \\
p_{i,3} &= p_3 = \mathbb{P}\big(H_1(k)\big) \mathbb{P}\big(\gamma_{_{mC}}(k)>\gamma_{req}\big) \\
p_{i,4} &= p_4 = \mathbb{P}\big(H_1(k)\big) \mathbb{P}\big(\gamma_{_{mC}}(k)<\gamma_{req}\big) \\
p_{i,5} &= p_5 = \mathbb{P}\big(H_2(k)\big) \mathbb{P}\big(\gamma_{_{MC}}(k)>\gamma_{req}\big) \\
p_{i,6} &= p_6 = \mathbb{P}\big(H_2(k)\big) \mathbb{P}\big(\gamma_{_{MC}}(k)<\gamma_{req}\big),
\end{split}
\end{equation}
where $\mathbb{P}(\gamma_d(k)<\gamma_{req}) = 1-e^{-\gamma_{req}/\E(\gamma_d(k))}$. $\mathbb{P}(H_1(k))=\mathbb{P}(H_1|H_0)+\mathbb{P}(H_1|H_1)+\mathbb{P}(H_1|H_2)$, where $\mathbb{P}(H_1|H_1) = P_{_{d,H_1}}$. Similarly, $\mathbb{P}(H_2(k))=\mathbb{P}(H_2|H_0)+\mathbb{P}(H_2|H_1)+\mathbb{P}(H_2|H_2)$, where $\mathbb{P}(H_2|H_2) = P_{_{d,H_2}}$. Note that $\gamma_{_{mC}}(k)$ and $\gamma_{_{MC}}(k)$ are also exponentially distributed random variables (R.V.) (because the minimum of two exponentially distributed R.V.s is also an exponential R.V.). Thus, $\mathbb{P}(\gamma_{_{mC}}(k)>\gamma_{req}) = e^{-\gamma_{req}/\E(\gamma_{_{mC}}(k))}$, where $\E(\gamma_{_{mC}}(k))=\frac{\E[\gamma^{^{mC}}_{_{ul}}]\E[\gamma^{^{mC}}_{_{dl}}]}{\E[\gamma^{^{mC}}_{_{ul}}]+\E[\gamma^{^{mC}}_{_{dl}}]}$, with $\E[\gamma^{^{mC}}_{_{ul}}]=\frac{\bar{P}}{1+\bar{\alpha}\bar{P}_{_{mC}}^{\beta}}$ and $\E[\gamma^{^{mC}}_{_{dl}}]=\frac{\bar{P}_{_{mC}}}{L^{^{mC}}_{_{dl}}(k) N_0}$. Finally, $\mathbb{P}(\gamma_{_{mC}}(k)<\gamma_{req}) = 1-e^{-\gamma_{req}/\E(\gamma_{_{mC}}(k))}$. Similarly, one can find $\mathbb{P}(\gamma_{_{MC}}(k)>\gamma_{req})$ and $\mathbb{P}(\gamma_{_{MC}}(k)<\gamma_{req})$ using the same framework, which turns out to be $e^{-\gamma_{req}/\E(\gamma_{_{MC}}(k))}$ and $e^{-\gamma_{req}/\E(\gamma_{_{MC}}(k))}$, respectively. With this, each row of $\mathbf{P}_o$ becomes: $\mathbf{p}_{o,i}=[p_{o,1}, p_{o,2}, p_{o,3}, p_{o,4}, p_{o,5}, p_{o,6}]$. Note that, due to identical rows, $\mathbf{P}_o$ has rank 1.
\begin{remark}
The mC-D2D and MC-D2D modes transfer data from $D_T$ to $D_R$ using a two-hop communication link. This implies two queues in the network; one at $D_T$ and the other at the BS. However, this work assumes that both BSs ($BS_{_{mC}}$ and $BS_{_{MC}}$) have infinite-sized queues, know perfect CSI ($Z_{dl}^{^{mC}}$ and $Z_{dl}^{^{MC}}$), and their average transmit powers ($\bar{P}_{_{mC}}$ and $\bar{P}_{_{MC}}$) are greater than the average transmit power of $D_T$ ($\bar{P}$). Therefore, the problem of queue overflow does not occur at either of the BS.
\end{remark}

\subsection{Markov Chain Modelling of Underlay-D2D}
In the underlay-D2D scenario, $D_T$ and $D_R$ reuses the cellular user's resources; hence, they experiences interference from $U_T$. Therefore, to compute the channel capacities $C^{u}_{d}(k)$, $C^{u}_{_{mC}}(k)$, and $C^{u}_{_{MC}}(k)$, we calculate the signal-to-interference-and-noise ratio (SINR) in each communication mode, which is defined as $\Gamma_{d}(k)$, $\Gamma_{_{mC}}(k)$, and $\Gamma_{_{MC}}(k)$. The SINR for the direct-D2D mode can be calculated as : $\Gamma_d(k)=\frac{\bar{P} Z_d(k)/L_d}{I_d+N_0}$, where $I_d=\frac{\bar{P}_{U_T} Z_{U_T,D_R}}{L_{U_T,D_R}}$. $\bar{P}_{U_T}$ is the average transmit power of $U_T$, and pathloss and the channel coefficients between $U_T$ and $D_R$ are $L_{U_T,D_R}$ and $Z_{U_T,D_R}$, respectively. The SINRs of UL and DL of the mC-D2D mode can be written as $\Gamma^{^{mC}}_{_{ul}}(k)=\frac{\bar{P} Z^{^{mC}}_{_{ul}}(k)/L^{^{mC}}_{_{ul}}}{I^{^{mC}}_{_{ul}}+N_0+\alpha\bar{P}_{_{mC}}^{\beta}}$ and $\Gamma^{^{mC}}_{_{dl}}(k)=\frac{\bar{P}_{mC} Z^{^{mC}}_{_{ul}}(k)/L^{^{mC}}_{_{ul}}(k)}{I_{d}+N_0}$, where $I^{^{mC}}_{_{ul}}=\frac{\bar{P}_{U_T} Z_{U_T,mC}(k)}{L_{U_T,mC}}$. Here, $Z_{U_T,mC}(k)$ and $L_{U_T,mC}(k)$ represent the channel coefficient and pathloss between $U_T$ and $BS_{_{mC}}$ in time slot $k$, respectively. Similarly, the SINRs on UL and DL in MC-D2D mode are $\Gamma^{^{MC}}_{_{ul}}(k)=\frac{\bar{P} Z^{^{MC}}_{_{ul}}(k)/L^{^{MC}}_{_{ul}}}{I^{^{MC}}_{_{ul}}+N_0+\alpha\bar{P}_{_{MC}}^{\beta}}$ and $\Gamma^{^{MC}}_{_{dl}}(k)=\frac{\bar{P}_{MC} Z^{^{MC}}_{_{ul}}(k)/L^{^{MC}}_{_{ul}}(k)}{I_{d}+N_0}$, respectively, where $I^{^{MC}}_{_{ul}}=\frac{\bar{P}_{U_T} Z_{U_T,MC}(k)}{L_{U_T,MC}}$. Note that the underlay scenario requires re-computation of six probabilities given in \eqref{eq:p11} and \eqref{eq:p2p3p4}. To do so, we consider an interference-limited scenario, whereby, by neglecting noise, we obtain signal-to-interference (SIR) expressions for all three communication modes. For the case of direct-D2D mode, the SIR expression would be: $\Upsilon_d = \frac{\Psi_d}{I_d}$, where $\Psi_d = \frac{\bar{P}Z_d(k)}{L_d}$. Observe that $\Psi_d \sim \exp(\frac{L_d}{\bar{P}})$ and $I_d \sim \exp(\frac{L_{U_T,D_R}}{\bar{P_{U_T}}})$. Then, the outage probability for the direct-D2D mode becomes
\begin{equation}\label{outage_prob_underlay1}
  \begin{split}
  \mathbb{P}(\Upsilon_d (k) < \gamma_{req})&=\frac{L_d/\bar{P}}{\frac{L_d/\bar{P}+L_{U_T,D_R}/P_{U_T}}{\gamma_{_{req}}}}\\
  &=\frac{L_d\gamma_{_{req}}P_{U_T}}{L_dP_{U_T}+L_{U_T,D_R}\bar{P}}.
  \end{split}
\end{equation}
Similarly, the probability $\mathbb{P}(\Upsilon_d (k) < \gamma_{req})$ becomes
\begin{equation}\label{outage_prob_underlay2}
  \mathbb{P}(\Upsilon_d (k) > \gamma_{req}) = 1-\frac{L_d\gamma_{_{req}}P_{U_T}}{L_dP_{U_T}+L_{U_T,D_R}\bar{P}}.
\end{equation}
The probabilities in (\ref{outage_prob_underlay1}) and (\ref{outage_prob_underlay2}) allow us to compute $p_{u,1}$ and $p_{u,2}$. Now, for the mC-D2D mode, let $\Upsilon_{_{mC}} = \min\{\Upsilon^{^{mC}}_{_{ul}}, \Upsilon^{^{mC}}_{_{dl}}\}$, where $\Upsilon^{^{mC}}_{_{ul}} = \Psi_{_{ul}}^{^{mC}}/I_{_{ul}}^{^{mC}}$ and $\Upsilon^{^{mC}}_{_{dl}} = \Psi_{_{dl}}^{^{mC}}/I_d$ are the SIR expressions for $D_T \to BS_{_{mC}}$ and $ BS_{_{mC}} \to D_R$ links, respectively, and where $\Psi_{_{ul}}^{^{mC}} = \bar{P}Z_{_{ul}}^{^{mC}}/L_{_{ul}}^{^{mC}}$ and $\Psi_{_{dl}}^{^{mC}} = \bar{P}_{_{mC}}Z_{_{dl}}^{^{mC}}/L_{_{dl}}^{^{mC}}$. Also, observe that $P_{_{ul}}^{^{mC}} \sim \exp(L_{_{ul}}^{^{mC}}/\bar{P})$, $P_{_{dl}}^{^{mC}} \sim \exp(L_{_{dl}}^{^{mC}}/\bar{P}_{_{mC}})$, and $I_{_{ul}}^{^{mC}} \sim \exp(L_{_{U_T,mC}}/\bar{P}_{_{mC}})$. Because $\Upsilon_{_{ul}}^{^{mC}}$ and $\Upsilon_{_{dl}}^{^{mC}}$ are independent R.V., the outage probability for mC-D2D mode becomes
\begin{equation}\label{outage_prob_underlay3}
    \begin{split}
    &\mathbb{P}(\Upsilon_{_{mC}}(k) < \gamma_{_{req}})\\
    &= \frac{L_{_{ul}}^{^{mC}}/\bar{P}}{\frac{L_{_{ul}}^{^{mC}}/\bar{P}+L_{_{dl}}^{^{mC}}/\bar{P}_{_{mC}}}{\gamma_{_{req}}}} + \frac{L_{_{U_T,mC}}/\bar{P}_{_{mC}}}{\frac{L_{_{U_T,mC}}/\bar{P}_{_{mC}}+ L_{U_T,D_R}/\bar{P}_{U_T}}{\gamma_{_{req}}}}\\
    &- \frac{L_{_{ul}}^{^{mC}}/\bar{P}}{\frac{L_{_{ul}}^{^{mC}}/\bar{P}+L_{_{dl}}^{^{mC}}/\bar{P}_{_{mC}}}{\gamma_{_{req}}}} \times \frac{L_{_{U_T,mC}}/\bar{P}_{_{mC}}}{\frac{L_{_{U_T,mC}}/\bar{P}_{_{mC}}+ L_{U_T,D_R}/\bar{P}_{U_T}}{\gamma_{_{req}}}}\\
    &=\frac{\gamma_{_{req}}\big[L^{^{mC}}_{_{ul}}\bar{P}_{_{mC}}(-\gamma_{_{req}}\bar{P}_{U_T} + \bar{P}_{_{mC}}+2\bar{P}_{_{mC}} ) + L_{_{dl}}^{^{mC}}\bar{P}\bar{P}_{U_T} \big]}{(\bar{P}_{U_T} + \bar{P}_{_{mC}})(L_{_{dl}}^{^{mC}}\bar{P}+L_{_{ul}}^{^{mC}}\bar{P}_{_{mC}})}.
  \end{split}
\end{equation}
Similarly, the probability $\mathbb{P}(\Upsilon_{_{mC}}(k) > \gamma_{_{req}})$ becomes
\begin{equation}\label{outage_prob_underlay4}
\begin{split}
&\mathbb{P}(\Upsilon_{_{mC}}(k) > \gamma_{_{req}})\\
&=1-\frac{\gamma_{_{req}}\big[L^{^{mC}}_{_{ul}}\bar{P}_{_{mC}}(-\gamma_{_{req}}\bar{P}_{U_T} + \bar{P}_{_{mC}}+2\bar{P}_{_{mC}} ) + L_{_{dl}}^{^{mC}}\bar{P}\bar{P}_{U_T} \big]}{(\bar{P}_{U_T} + \bar{P}_{_{mC}})(L_{_{dl}}^{^{mC}}\bar{P}+L_{_{ul}}^{^{mC}}\bar{P}_{_{mC}})}.
\end{split}
\end{equation}
The probabilities in (\ref{outage_prob_underlay3}) and (\ref{outage_prob_underlay4}) allow us to compute $p_{u,3}$ and $p_{u,4}$. For the MC-D2D mode, let $\Upsilon_{_{MC}} = \min\{\Upsilon^{^{MC}}_{_{ul}},\Upsilon^{^{MC}}_{_{dl}}\}$, where $\Upsilon^{^{MC}}_{_{ul}} = \Psi_{_{ul}}^{^{MC}}/I_{_{ul}}^{^{MC}}$ and $\Upsilon^{^{MC}}_{_{dl}} = \Psi_{_{dl}}^{^{MC}}/I_d$ are the SIR expressions for $D_T \to BS_{_{MC}}$ and $ BS_{_{MC}} \to D_R$ links, respectively. Here, $\Psi_{_{ul}}^{^{MC}} = \bar{P}Z_{_{ul}}^{^{MC}}/L_{_{ul}}^{^{MC}}$ and $\Psi_{_{dl}}^{^{MC}} = \bar{P}_{_{MC}}Z_{_{dl}}^{^{MC}}/L_{_{dl}}^{^{MC}}$. Also, observe that $P_{_{ul}}^{^{MC}} \sim \exp(L_{_{ul}}^{^{MC}}/\bar{P})$, $P_{_{dl}}^{^{MC}} \sim \exp(L_{_{dl}}^{^{MC}}/\bar{P}_{_{MC}})$, and $I_{_{ul}}^{^{MC}} \sim \exp(L_{_{U_T,mC}}/\bar{P}_{_{MC}})$. Similar to the case of the mC-D2D mode, $\Upsilon^{^{MC}}_{_{ul}}$ and $\Upsilon^{^{MC}}_{_{dl}}$ are independent random variables; therefore, $\mathbb{P}(\Upsilon_{_{MC}}(k) < \gamma_{_{req}})$ and $\mathbb{P}(\Upsilon_{_{MC}}(k) > \gamma_{_{req}})$ can be calculated by substituting $L_{_{ul}}^{^{MC}}$, $L_{_{dl}}^{^{MC}}$, and $\bar{P}_{_{MC}}$ into (\ref{outage_prob_underlay3}) and (\ref{outage_prob_underlay4}), respectively. These probabilities will then allow us to compute $p_{u,5}$ and $p_{u,6}$. By using the probabilities found above, we can find the state transition probability matrix for underlay D2D ($\mathbf{P}_u$). Similar to $\mathbf{P}_o$, $\mathbf{P}_u$ is also of unit rank 1, with each row $\mathbf{p}_{u,i}=[p_{u,1}, p_{u,2}, p_{u,3}, p_{u,4}, p_{u,5}, p_{u,6}]$.

\subsection{Effective Capacity of HARQ-enabled D2D}
In our analysis, we use HARQ for retransmission of the packet. In HARQ, each data packet is encoded into $M$ codeword blocks, and $M$ defines the maximum number of the allowed retransmissions of a packet, which is adjustable according to the reliability and delay requirements of the system \cite{yadav2017energy}.
Let us consider a transmission period $T$ containing $M$ codewords/fading blocks, with $l$ as the size of each fading block. In each transmission period, a codeword is transmitted; if $D_R$ decodes the codeword successfully, it sends an ACK, and the transmission period ends. Contrarily, if decoding fails at $D_R$, a NACK is sent to $D_T$; then, $D_T$ retransmits the packet with a new set of parity bits (codeword). This process continues until the packet is decoded successfully at $D_R$ or until the maximum limit of the retransmissions ($M$) is reached. Note that in HARQ, when $D_R$ decodes the received packet at the $m^{\text{th}}$ retransmission attempt (using $m$ number of codewords), it means that $m-1$ number of trials have finished and were unsuccessful. If $D_R$ fails to decode a packet on the $M^{\text{th}}$ retransmission attempt, an outage occurs. At that point, $D_T$ has two options: either delete that packet from the queue or reduce the priority of that packet and transmit the next packet with the highest priority. In the second option, the failed packet will then be transmitted when its priority becomes highest. We have modelled this scenario into two queue models. In model 1 ($n_1$), if a packet is not successfully decoded by $D_R$ even after the deadline occurs ($M$ number of unsuccessful attempts), then the packet's priority is reduced and the packet possessing the highest priority is transmitted in the following transmission period. In model 2 ($n_2$), the packet is deleted from $D_T$'s queue if not successfully decoded by $D_R$ after $M$ number of retransmission attempts.

The EC of HARQ-enabled D2D communication under the assumption of constant arrival ($a$) and transmission rates ($r$), given the QoS exponent $\theta$ and the specified retransmission constraint $M$, is given as follows \cite{larsson2016effective},
\begin{equation}\label{EC_main}
  EC^{^\text{HARQ}}_{n_j} = \frac{-1}{\theta} \log_e (\lambda_{n_j} +),
\end{equation}
where $\lambda_{n_j} +$ = $\max \{|\lambda_{1,n_j}|,|\lambda_{2,n_j}|,\dots,|\lambda_{M,n_j}|\}$ is the spectral radius of $\mathbf{B}_{n_j}$ and $j\in \{1,2\}$. $\mathbf{B}_{n_j}$ is a block-companion matrix of size $M \times M$ and is defined as,
 \begin{equation}\label{B_sub1}
\mathbf{B}_{n_j}=
\begin{bmatrix}
  b_{1,n_j} & b_{2,n_j} & \dots &b_{M-1,n_j}& b_{M,n_j} \\
 1 & 0 & \dots & 0 & 0 \\
  0 & 1 & \dots & 0 & 0 \\
  \vdots & \vdots & \ddots & \vdots & \vdots\\
  0 & 0 & \dots & 1 & 0
\end{bmatrix}.
\end{equation}
To find the entries of the matrix $\mathbf{B}_{n_j}$, first we have to find the decoding error and successful decoding probabilities at $D_R$ in each queue model. According to the finite block length coding rate model \cite{polyanskiy2009dispersion}, the decoding error probability of the $m^{\text{th}}$ transmission attempt in direct-D2D mode ($\zeta^{d}_{m}(Z)$), mC-D2D mode ($\zeta^{^{mC}}_{m}(Z)$), and MC-D2D mode ($\zeta^{^{MC}}_{m}(Z)$) can be written as \cite{hu2020throughput}
\begin{subequations}\label{decoding_error_d}
\begin{align}
  \zeta^{d}_{m}(Z) &= Q \bigg( \frac{\sum_{k=1}^{m}\log_2(1+\gamma_d(k))+{\log(ml)/l} - r}{\log_2e\sqrt{\sum_{k=1}^{m}\frac{(2+\gamma_d(k))\gamma_d(k)}{l(\gamma_d(k) + 1)^2}}}\bigg)\\
  \zeta^{^{mC}}_{m}(Z) &= Q \bigg( \frac{\sum_{k=1}^{m}\log_2(1+\gamma_{_{mC}}(k))+{\log(ml)/l} - r}{\log_2e\sqrt{\sum_{k=1}^{m}\frac{(2+\gamma_{_{mC}}(k))\gamma_{_{mC}}(k)}{l(\gamma_{_{mC}}(k) + 1)^2}}}\bigg)\\
  \zeta^{^{MC}}_{m}(Z) &= Q \bigg( \frac{\sum_{k=1}^{m}\log_2(1+\gamma_{_{MC}}(k))+{\log(ml)/l} - r}{\log_2e\sqrt{\sum_{k=1}^{m}\frac{(2+\gamma_{_{MC}}(k))\gamma_{_{MC}}(k)}{l(\gamma_{_{MC}}(k) + 1)^2}}}\bigg).
\end{align}
\end{subequations}
Where $\gamma_d(k)$, $\gamma_{_{mC}}(k)$, and $\gamma_{_{MC}}(k)$ are the SNR of the direct-D2D, mC-D2D, and MC-D2D modes, respectively. Let us define $P_{t,\nu,n_j}$ as the probability of $\nu$, the number of removed packets from $D_T$'s queue, for the queue model, $j$, in time period, $t$. We know from the deadline constraint that $1\leq t \leq M$ and $\nu \in \{0,1\}$ (considering that only one packet is being transmitted in one transmission period).

\paragraph*{\textbf{Queue Model 1 $(n_1)$}} In $n_1$, $\nu = 0$ when outage occurs ($t=M$); therefore, we can say that $P_{t,0,n_1}$ is the probability that no successful decoding happens at $D_R$ when $M$ is reached. Contrarily, $P_{t,1,n_1}$ represents the probability that a transmission period ended successfully in the $t^{\text{th}}$ time block. From here, we have the following:
\begin{equation}\label{eq:P_t_0_m1}
  P_{t,0,n_1} =
  \begin{cases}
    \left.\begin{aligned}
        0, \quad &t<M\\
        \varepsilon_d,\quad &t=M
       \end{aligned}
 \;\right\}
  \quad \text{direct-D2D mode} \\
    \left.\begin{aligned}
        0, \quad &t<M\\
        \varepsilon_{_{mC}}, \quad &t=M
       \end{aligned}
 \;\right\}
  \quad \text{mC-D2D mode}\\
  \left.\begin{aligned}
        0, \quad &t<M\\
        \varepsilon_{_{MC}},\quad &t=M
       \end{aligned}
 \;\right\}
  \quad \text{MC-D2D mode}
  \end{cases}
\end{equation}
where $\varepsilon_d$, $\varepsilon_{_{mC}}$, and $\varepsilon_{_{MC}}$ are the outage probabilities in direct-D2D, mC-D2D, and MC-D2D modes, respectively. These outage probabilities can be defined as $\E_z[\zeta^{^{d}}_{M}]$, $ \E_z[\zeta^{^{mC}}_{M}]$, and $ \E_z[\zeta^{^{MC}}_{M}]$, respectively. The probability that $D_R$ successfully decodes the packet in $t^{\text{th}}$ time block is equal to the probability of $D_R$ decoding the packet within $t$ time blocks minus the probability of $D_R$ decoding the packet within $t-1$ time blocks. Therefore, $P_{t,1,n_1}$ can be defined as
\begin{equation}
	\label{eq:P_t_1_m1}
P_{t,1,n_1}=
	 \begin{cases}   \E_z[\zeta^{^{d}}_{t-1}]-\E_z[\zeta^{^{d}}_{t}], & \text{direct-D2D mode}\\
                     \E_z[\zeta^{^{mC}}_{t-1}]-\E_z[\zeta^{^{mC}}_{t}], & \text{mC-D2D mode}\\
                     \E_z[\zeta^{^{MC}}_{t-1}]-\E_z[\zeta^{^{MC}}_{t}], & \text{MC-D2D mode}
 \end{cases}.
\end{equation}

\paragraph*{\textbf{Queue Model 2 $(n_2)$}} In $n_2$, $\nu = 1$ due to the fact that a packet surely leaves $D_T$'s queue as each transmission period ends. It is either because of the successful decoding of the packet at $D_R$ or because of the packet dropped by $D_T$'s queue when $M$ is reached. In the $n_2$ model, $t<M$ corresponds to the successful transmission of the packet, as it also did in the $n_1$ model. In the $n_2$ model, $t=M$ corresponds to two cases. The first case is when $D_R$ decodes the packet successfully in the $M^{\text{th}}$ time block. The second case is when an outage occurs, consequently dropping the packet from $D_T$'s queue. Therefore, we have the following cases
\begin{equation}\label{eq:P_t_1_m2}
  P_{t,1,n_2} =
  \begin{cases}
    \left.\begin{aligned}
        \E_z[\zeta^{^{d}}_{t-1}]-\E_z[\zeta^{^{d}}_{t}], \quad &t<M\\
        \E_z[\zeta^{^{d}}_{M-1}],\quad &t=M
       \end{aligned}
 \;\right\}
  \quad \text{direct-D2D mode} \\
    \left.\begin{aligned}
        \E_z[\zeta^{^{mC}}_{t-1}]-\E_z[\zeta^{^{mC}}_{t}], \quad &t<M\\
        \E_z[\zeta^{^{mC}}_{M-1}], \quad &t=M
       \end{aligned}
 \;\right\}
  \quad \text{mC-D2D mode}\\
  \left.\begin{aligned}
        \E_z[\zeta^{^{MC}}_{t-1}]-\E_z[\zeta^{^{MC}}_{t}], \quad &t<M\\
        \E_z[\zeta^{^{MC}}_{M-1}],\quad &t=M
       \end{aligned}
 \;\right\}
  \quad \text{MC-D2D mode} \\
  \end{cases}
\end{equation}
For the case of $t=M$, we use $P_{t,1,n_2} = \E_z[\zeta_{M-1}]-\E_z[\zeta_M] + \varepsilon$, where $\E_z[\zeta_M]=\varepsilon$.

Now, to find the entries of the block companion matrix $\mathbf{B}_{n_j}$, we utilize the results from (\ref{eq:P_t_0_m1}), (\ref{eq:P_t_1_m1}), and (\ref{eq:P_t_1_m2}); consequently, we obtain the following
\begin{equation}
	\label{eq:b_k_nj}
b_{k,n_j}=
	 \begin{cases}   \mathbf{q}_{_{1}} \mathbf{\Phi}(-\theta)\mathbf{p}_{i}^{\intercal}, &k=1\\
                     \mathbf{q}_{_{2}} \mathbf{\Phi}(-\theta)\mathbf{p}_{i}^{\intercal}, &2\leq k \leq M-1 \\
                     \mathbf{q}_{_{3}} \mathbf{\Phi}(-\theta)\mathbf{p}_{i}^{\intercal}+\varepsilon_{_{ac}}, &k=M \; \text{and}\; j=1 \\
                     \mathbf{q}_{_{4}} \mathbf{\Phi}(-\theta)\mathbf{p}_{i}^{\intercal}, &k=M\; \text{and}\; j=2, \end{cases}
\end{equation}
where $\mathbf{q}_{_{1}} = \big[1-\E_z[\zeta^{^{d}}_{1}], 1, 1-\E_z[\zeta^{^{mC}}_{1}], 1, 1-\E_z[\zeta^{^{MC}}_{1}], 1\big]$, $\mathbf{q}_{_{2}} = \big[\E_z[\zeta^{^{d}}_{k-1}]-\E_z[\zeta^{^{d}}_{k}], 1, \E_z[\zeta^{^{mC}}_{k-1}]-\E_z[\zeta^{^{mC}}_{k}], 1, \E_z[\zeta^{^{MC}}_{k-1}]-\E_z[\zeta^{^{MC}}_{k}], 1\big]$, $\mathbf{q}_{_{3}} = \big[\E_z[\zeta^{^{d}}_{M-1}]-\varepsilon_d, 1, \E_z[\zeta^{^{mC}}_{M-1}]-\varepsilon_{_{mC}}, 1, \E_z[\zeta^{^{MC}}_{M-1}]-\varepsilon_{_{MC}}, 1\big]$, and $\mathbf{q}_{_{4}} = \big[\E_z[\zeta^{^{d}}_{M-1}], 1, \E_z[\zeta^{^{mC}}_{M-1}], 1, \E_z[\zeta^{^{MC}}_{M-1}], 1\big]$. 
$\varepsilon_{_{ac}} = \varepsilon_d+\varepsilon_{_{mC}}+\varepsilon_{_{MC}}$ is the accumulative outage probability, $\mathbf{p}_i$ is a vector containing all the state transition probabilities (due to unit rank), $\mathbf{p}_i = [p_1, p_2, p_3, p_4, p_5, p_6]$, and $\mathbf{\Phi}(\theta)$ is the diagonal matrix containing the MGF of the processes in the six states ($s_1$, $s_2$, $s_3$, $s_4$, $s_5$, $s_6$). Because $S(k) = r$ for states $s_1$, $s_3$, and $s_5$ (ON states) and $S(k) = 0$ for states $s_2$, $s_4$, and $s_6$ (OFF states). Therefore, the MGFs for states $s_1$, $s_2$, $s_3$, $s_4$, $s_5$, and $s_6$ become $e^{lr\theta}$, 1, $e^{lr\theta}$, 1, $e^{lr\theta}$, and 1, respectively. Thus, $\mathbf{\Phi}(-\theta)$ can be expressed as $\mathbf{\Phi}(-\theta) = \text{diag}[e^{-lr\theta}, 1, e^{-lr\theta}, 1, e^{-lr\theta}, 1]$. By substituting these values in \eqref{eq:b_k_nj} and by setting a limit on the packet retransmissions, one can find entries of the block companion matrix $B_{n_j}$. Further, by calculating the spectral radius of $B_{n_j}$, one can find the EC of HARQ-enabled D2D communication for both queue models. Next, we investigate a special case of HARQ by adjusting the retransmission limit to 2, and provide its statistical QoS analysis.
\begin{figure}[ht]
\begin{center}
	\includegraphics[width=3in]{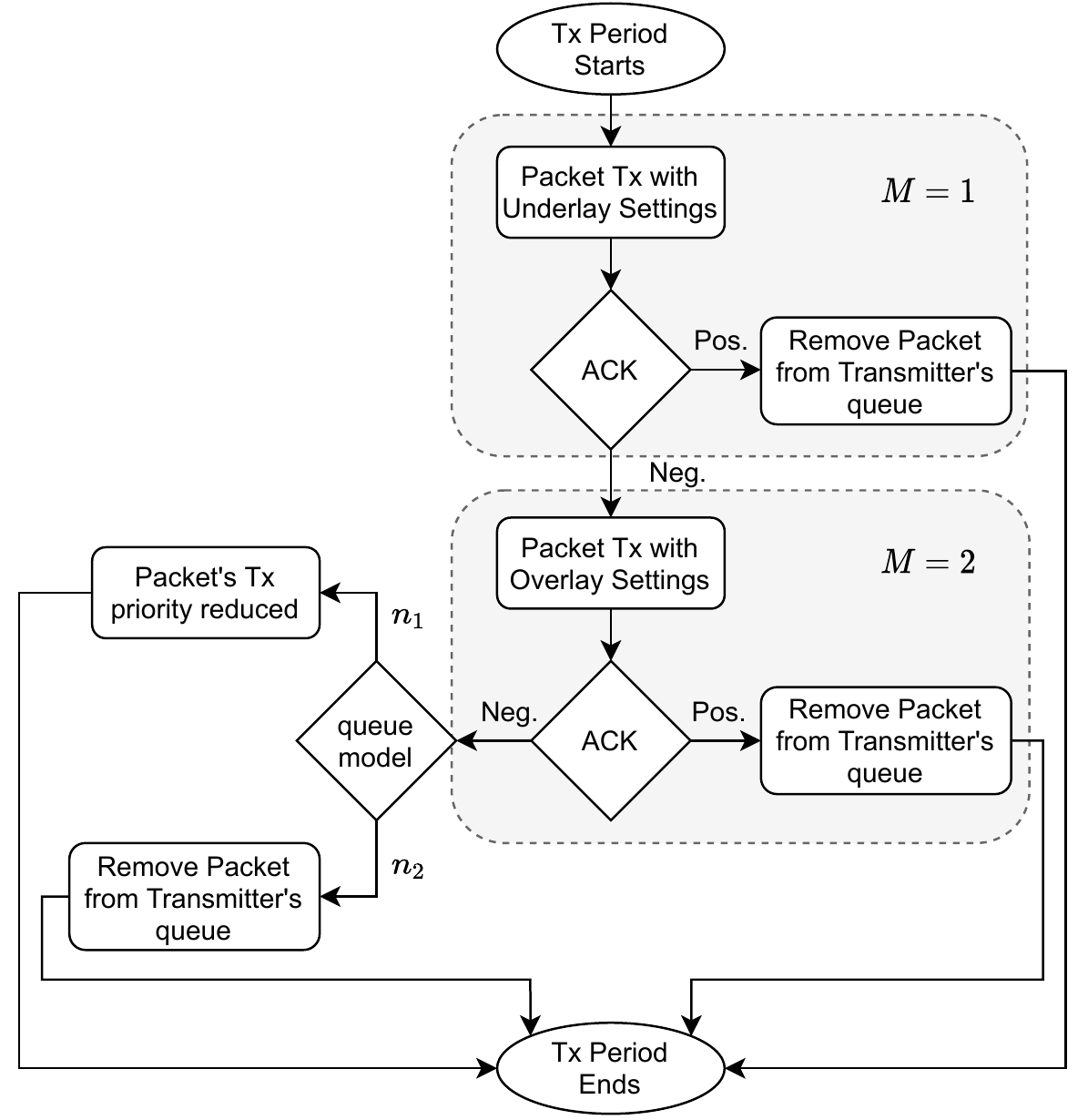}
\caption{Flow diagram of truncated HARQ-enabled D2D communication.}
\label{flow_chart}
\end{center}
\end{figure}

\subsection{Effective Capacity of Truncated HARQ-enabled D2D}
In this subsection, we discuss a special case of HARQ (truncated HARQ \cite{malkamaki2000performance}) and also provide closed-form expression for the EC of truncated HARQ-enabled D2D communication. We restrict the maximum number of packet transmissions in a transmission period to its lowest value, which is $M=2$. In this case, $D_T$ first transmits a packet using underlay settings by reusing the cellular user's channel. If the packet fails to be decoded at $D_R$, then the packet is retransmitted using overlay settings in the same transmission period, as shown in Fig. \ref{flow_chart}. This way, we can achieve higher reliability by utilizing less network resources. For $M=2$, the block companion matrix $\mathbf{B}_{n_j}$ would become
\begin{equation}\label{B_special}
\mathbf{B}_{n_j}=
\begin{bmatrix}
  b_{1,n_j} & b_{2,n_j}\\
 1 & 0
\end{bmatrix},
\end{equation}
and the corresponding characteristic equation is $\lambda_{n_j}^2 - \lambda_{n_j}(b_{1,n_j})-b_{2,n_j}=0$, with the largest positive root for queue model ${n_j}$
\begin{equation}\label{lambda_n1}
  \lambda_{n_j}+ = \frac{1}{2}\bigg(b_{1,n_j}+\sqrt{(b_{1,n_j})^2+4(b_{2,n_j})}\bigg).
\end{equation}

Now, to find the EC expressions for queue model $n_1$ and $n_2$, we have to find the largest positive roots of the corresponding block companion matrices. For the largest positive root for queue model $n_1$ ($\lambda_{n_1}+$), we have the following Lemma \ref{lem1}.

\begin{lemma}\label{lem1}
The largest positive root of the block companion matrix for queue model $n_1$ is given as,
\begin{equation*}\label{lambda_n1_final}
\begin{split}
  \lambda_{n_1}+ &= \frac{1}{2}\biggl(\big(e^{-lr\theta}\big[\varphi\big]+p^{\text{off}}_{u}\big)+\\
  &\sqrt{\big(e^{-lr\theta}\big[\varphi\big]+p^{\text{off}}_{u}\big)^2+4\big(e^{-lr\theta}\big[\vartheta\big]+p^{\text{off}}_{o}+\varepsilon_{_{ac}}\big)}\biggl).
\end{split}
\end{equation*}
Where $\varphi = p_{u,1}\big(\alpha_d\big)+p_{u,3}\big(\alpha_{_{mC}}\big)+p_{u,5}\big(\alpha_{_{MC}}\big)$, $\vartheta = p_{o,1}\big(\beta_d\big)+p_{o,3}\big(\beta_{_{mC}}\big)+p_{o,5}\big(\beta_{_{MC}}\big)$, $p^{\text{off}}_{u} = p_{u,2}+p_{u,4}+p_{u,6}$, and $p^{\text{off}}_{o} = p_{o,2}+p_{o,4}+p_{o,6}$
\end{lemma}
\begin{proof}
Given in Appendix \ref{prop2}.
\end{proof}

By using results from Lemma \ref{lem1} and solving \eqref{EC_main}, we can find the closed-form expression for the EC of truncated HARQ-enabled D2D for queue model $n_1$, which is
\begin{equation}\label{EC_n1_final}
\begin{split}
  EC^{^\text{HARQ}}_{n_1} &= \frac{-1}{\theta} \log_e \bigg\{\frac{1}{2}\biggl(\big(e^{-lr\theta}\big[\varphi\big]+p^{\text{off}}_{u}\big)+\\
  &\sqrt{\big(e^{-lr\theta}\big[\varphi\big]+p^{\text{off}}_{u}\big)^2+4\big(e^{-lr\theta}\big[\vartheta\big]+p^{\text{off}}_{o}+\varepsilon_{_{ac}}\big)}\biggl)\bigg\}.
\end{split}
\end{equation}

Similarly, for queue model $n_2$, the expression for the largest positive root ($\lambda_{n_2}+$) can be found by using the following Lemma \ref{lem2}.
\begin{lemma}\label{lem2}
The largest positive root of the block companion matrix for queue model $n_2$ is given as,
\begin{equation*}\label{lambda_n2_final}
\begin{split}
 \lambda_{n_2}+ = \frac{1}{2}\bigg(\big(e^{-lr\theta}&\big[\varphi\big]+p^{\text{off}}_{u}\big) +\\
  &\sqrt{\big(e^{-lr\theta}\big[\varphi\big]+p^{\text{off}}_{u}\big)^2 +4e^{-lr\theta}\big[\varrho\big]+p^{\text{off}}_{o}}\bigg).
\end{split}
\end{equation*}
Where $\varrho = p_{o,1}\E_z[\zeta^{^{d}}_{o,1}]+p_{o,3}\E_z[\zeta^{^{mC}}_{o,1}]+p_{o,5}\E_z[\zeta^{^{MC}}_{o,1}]$.
\end{lemma}
\begin{proof}
Given in Appendix \ref{prop3}.
\end{proof}

Now, by using results from Lemma \ref{lem2} and solving \eqref{EC_main}, we can find the closed-form expression for the EC of truncated HARQ-enabled D2D for the queue model $n_2$, which is
\begin{equation}\label{EC_n2_final}
  \begin{split}
  EC^{^\text{HARQ}}_{n_2} &= \frac{-1}{\theta} \log_e \bigg\{\frac{1}{2}\bigg(\big(e^{-lr\theta}\big[\varphi\big]+p^{\text{off}}_{u}\big) +\\
  &\sqrt{\big(e^{-lr\theta}\big[\varphi\big]+p^{\text{off}}_{u}\big)^2 +4e^{-lr\theta}\big[\varrho\big]+p^{\text{off}}_{o}}\bigg)\bigg\}.
\end{split}
\end{equation}

We provide numerical investigation and insights of these EC expressions for both of the queue models in Section V.
\subsection{Optimal Transmission Rate}
As discussed above, we assume that CSIT is not available; therefore, the transmitting device sends data using a fixed transmission rate. To achieve the maximum EC, it is essential to transmit data using an optimal transmission rate. Therefore, in this section, we find the optimized transmission rates for $n_1$ and $n_2$ models that maximize the EC in respective queue models. These optimal transmission rates can be written as $r_{n_j}^{*} = \arg \max_{r_{n_j}>0}EC^{^{\text{HARQ}}}_{n_j}$. For $n_1$ model, it becomes
\begin{equation}\label{optimized_r_n1_first}
\begin{split}
  r_{n_1}^{*} &= \arg \max_{r_{n_1}>0}\frac{-1}{\theta} \log_e \bigg\{\frac{1}{2}\biggl(\big(e^{-lr_{n_1}\theta}\big[\varphi\big]+p^{\text{off}}_{u}\big)+\\
  &\sqrt{\big(e^{-lr_{n_1}\theta}\big[\varphi\big]+p^{\text{off}}_{u}\big)^2+4\big(e^{-lr_{n_1}\theta}\big[\vartheta\big]+p^{\text{off}}_{o}+\varepsilon_{_{ac}}\big)}\biggl)\bigg\}.
\end{split}
\end{equation}
Equivalently, we can write
\begin{equation}\label{optimized_r_n1_first}
\begin{split}
  r_{n_1}^{*} &= \arg \min_{r_{n_1}>0}\bigg\{\big(e^{-lr_{n_1}\theta}\big[\varphi\big]+p^{\text{off}}_{u}\big)+\\
  &\sqrt{\big(e^{-lr_{n_1}\theta}\big[\varphi\big]+p^{\text{off}}_{u}\big)^2+4\big(e^{-lr_{n_1}\theta}\big[\vartheta\big]+p^{\text{off}}_{o}+\varepsilon_{_{ac}}\big)}\bigg\}.
\end{split}
\end{equation}
From Table \ref{states}, we can see that the transmission is only possible in states $s_1$, $s_3$, and $s_5$ and that no transmission occurs during states $s_2$, $s_4$, and $s_6$. Therefore, the transmission probabilities $p_2$, $p_4$, and $p_6$, in both overlay and underlay scenarios, are irrelevant when optimizing (\ref{optimized_r_n1_first}) with respect to $r_{n_1}$. By discarding the irrelevant terms, the final optimization problem becomes
\begin{equation}\label{optimized_r_n1_final}
\begin{split}
  r_{n_1}^{*} = \arg \min_{r_{n_1}>0} \bigg\{&e^{-lr_{n_1}\theta}[\varphi] +\\
  &\sqrt{\big(e^{-lr_{n_1}\theta}[\varphi]\big)^2+4\big(e^{-lr_{n_1}\theta}[\vartheta]+\varepsilon_{_{ac}}\big)}\bigg\}.
\end{split}
\end{equation}
Let $F = e^{-lr_{n_1}\theta}[\varphi] + \sqrt{(e^{-lr_{n_1}\theta}[\varphi])^2+4(e^{-lr_{n_1}\theta}[\vartheta]+\varepsilon_{_{ac}})}$ be the cost function. Because $F$ is a convex function \cite{brychkov2012some}, we can find its closed-form by taking the derivative with respect to $r_{n_1}$. By taking the derivative of $F$ and by employing the chain rule and the sum/difference rule, we obtain the following result
\begin{equation}\label{gradient}
\frac{\partial F}{\partial r_{n_1}} = -l\theta e^{-lr_{n_1}\theta}[\varphi]-\frac{l\theta e^{-2lr_{n_1}\theta}([\varphi]+2[\vartheta]e^{lr_{n_1}\theta})}{\sqrt{(e^{-lr_{n_1}\theta}[\varphi])^2+4(e^{-lr_{n_1}\theta}[\vartheta]+\varepsilon_{_{ac}})}}.
\end{equation}
Now, to find the closed-form expression, we set $\frac{\partial F}{\partial r_{n_1}} = 0$. Consequently, we obtain
\begin{equation}\label{closed_form}
 l\theta e^{-lr_{n_1}\theta}[\varphi]= -\frac{l\theta e^{-2lr_{n_1}\theta}([\varphi]+2[\vartheta]e^{lr_{n_1}\theta})}{\sqrt{(e^{-lr_{n_1}\theta}[\varphi])^2+4(e^{-lr_{n_1}\theta}[\vartheta]+\varepsilon_{_{ac}})}}.
\end{equation}
Solving \eqref{closed_form} for $r_{n_1}$ requires a great deal of computation, and the computational complexity of the solution is very high. Therefore, we employ the iterative gradient decent (GD) method to determine the optimal transmission rate $r_{n_1}^*$. To control the convergence of the GD method, we have the following rule
\begin{equation}
    r_{n_1}(x) = r_{n_1}(x-1)-\Omega \nabla\big|_{r_{n_1}(x)},
\end{equation}
where $\Omega$ is the step-size, $x$ is the number of the iteration, and $\nabla$ is the gradient of $F$. This gradient can be written as $\nabla =\frac{\partial F}{\partial r_{n_1}}$ and is given in \eqref{gradient}.

Similarly, for $n_2$ model, the optimized transmission rate can be calculated using the following expression:
\begin{equation}\label{optimized_r_n2_final}
  r_{n_2}^{*} = \arg \min_{r_{n_2}>0} \bigg(e^{-lr_{n_2}\theta}\big[\varphi\big] + \sqrt{\big(e^{-lr_{n_2}\theta}\big[\varphi\big]\big)^2+4\big(e^{-lr_{n_2}\theta}\big[\varrho\big]\big)}\bigg).
\end{equation}
To solve \eqref{optimized_r_n2_final} and to find the optimal value of $r_{n_2}$, one can follow the same procedure used for $n_1$ model.
\begin{remark}
In our system model, the MC-BS performs the mode selection mechanism (to find the best mode for D2D communication) and executes the GD algorithm (to compute the optimal yet fixed transmission rates). The MC-BS then communicates the outcome of the mode selection and the optimal transmission rate to $D_T$ through the downlink control channel. Moreover, the mode selection and the optimal transmission rates have to be recomputed every time pathloss of the D2D link changes (due to the D2D users' mobility). The MC-BS performs these tasks because we assume that it has adequate resources to execute the GD algorithm. It also keeps track of the D2D users' mobility to decide when to recompute the optimal transmission rates.
\end{remark}
\section{Numerical Results}
In this section, we further investigate the EC of HARQ-enabled D2D communication and the impact of mode selection on the performance of the D2D link, and we provide simulation results to support our analysis.
\subsection{Simulation Setup}
We consider an MC of radius 500 m and an mC of radius 100 m in the MC's coverage area. Two pairs of user equipment are positioned in the coverage area of the mC using uniform distribution. One pair is referred to as the D2D pair ($D_T$ and $D_R$) and the other as the cellular user pair ($U_T$ and $U_R$). We use the pathloss as a sole-feature for mode selection, with the following pathloss model \cite{shah2019system}: L($d$)=128.1+37.6$\log_{10}(d)$. We use power class 1 devices at the transmitter and receiver, with their average transmit power set to be 27 dBm. The average transmit powers of MC-BS and mC-BS are 47 dBm and 37 dBm, respectively. We assume that the channels $D_T \to D_R$, $D_T \to BS_{_{mC}} \to D_R$, and $D_T \to BS_{_{MC}} \to D_R$ are Rayleigh fading channels and follow independent distributions. 
\subsection{Simulation Results}
\begin{figure}[ht]
\begin{center}
	\includegraphics[width=3in]{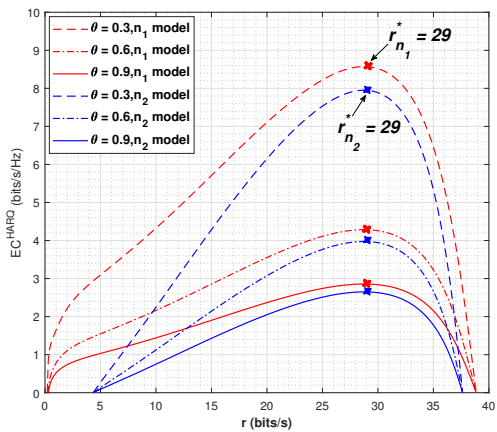}
\caption{$EC^{^{\text{HARQ}}}$ is a quasi-concave function of $r$; an exhaustive search to find the optimal $r$ for different values of QoS exponent $\theta$ ($M=2$).}
\label{EC_vs_r}
\end{center}
\end{figure}
Fig. \ref{EC_vs_r} presents a comprehensive search to determine the optimal value of the fixed transmission rate with a constant arrival rate. It can be seen that the EC of truncated HARQ-enabled D2D is a quasi-concave function of $r$ and that a globally optimal value of $r$ ($r_{n_1}^* = r_{n_2}^* = 29$) exists that maximizes the EC. This is because $r$ introduces a significant outage probability when it is too large. Consequently, a large amount of packet drop happens due to the deadline constraint. On the other hand, when $r$ is too small, it forces the departure rate low as well. In short, for large $r$, the decoding error probability is the bottleneck, and for small $r$, low departure rate is the bottleneck. From the figure, we can also see the impact of using different queue models. For instance, the $n_1$ queue model provides a higher EC on the optimal value of $r$ than does the $n_2$ queue model. This is because the unsuccessful packet is discarded when a deadline is approached in the $n_2$ model. On the other hand, in the $n_1$ model, packet transmission priority is reduced, rather than discarded, when it remains unsuccessful, even after the deadline is reached. Moreover, one can also see the impact of imposing strict QoS constraints on the EC; for instance, a lower EC is achieved at the optimal $r$ when stricter QoS constraints are imposed at $D_T$'s queue.

\begin{figure}[ht]
\begin{center}
	\includegraphics[width=3in]{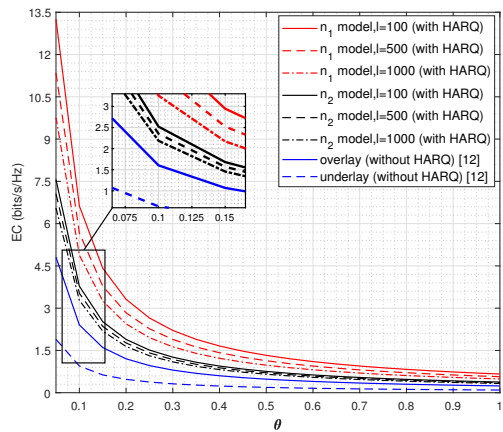}
\caption{The EC vs the QoS exponent $\theta$: a comparison of HARQ-enabled D2D communication with traditional D2D communication.}
\label{EC_vs_theta}
\end{center}
\end{figure}
Next, we investigate the effect of the QoS exponent on the EC of our proposed system model. Fig. \ref{EC_vs_theta} shows that the EC is a decreasing function of $\theta$. Specifically, the EC decreases exponentially fast for lower values of $\theta$. For higher values of $\theta$, this rate of decrease slows down and ultimately reaches zero when $\theta$ approaches 1. It also shows that our proposed scheme of truncated HARQ-enabled D2D outperforms other D2D schemes, such as overlay and underlay D2D. However, this gain over other D2D schemes decreases as stricter QoS constraints are imposed at $D_T$'s queue. Moreover, we also observe a significant performance loss when the finite blocklength ($l$) increases. This is because we consider a block-fading channel model; in such models, when the length of the fading block increases, the effect of slow-fading plays an important role. This occurs because slow-fading makes a strong attenuation last for a long time in delay-sensitive networks operating under statistical QoS constraints. This attenuation then causes an increase in the buffer overflow probability, which affects the performance of the system and results in reduced EC. Additionally, the results also show that the $n_1$ model with a large blocklength ($l=1000$) still outperforms the $n_2$ model with a small blocklength ($l=100$). It shows the efficacy of the $n_1$ model over the $n_2$ model in terms of performance but at the cost of more resources.\footnote{Note that the $n_1$ model requires comparatively more resources than the $n_2$ model because in the $n_1$ model, a packet is not discarded even after the retransmission deadline is reached, whereas in the $n_2$ model, a packet is discarded after the retransmission deadline is reached (which in our case occurs after two unsuccessful attempts). This phenomenon poses an extra burden on the resources available for D2D communication.}

\begin{figure}[ht]
\begin{center}
	\includegraphics[width=3in]{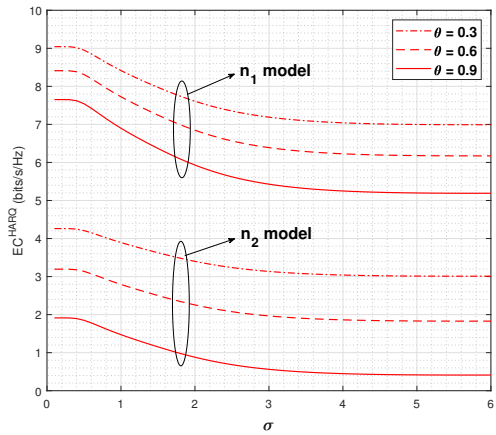}
\caption{Impact of the mode selection mechanism on the EC of HARQ-enabled D2D system: The EC of truncated HARQ-enabled D2D vs the standard deviation of the estimation error for $n_1$ and $n_2$ queue models.}
\label{EC_vs_sigma}
\end{center}
\end{figure}
Fig. \ref{EC_vs_sigma} presents the impact of our proposed mode selection on the EC of truncated HARQ-enabled D2D communication. The EC decreases initially with an increase in the standard deviation of the estimation error ($\sigma$) of pathloss measurements, and it becomes stable for $\sigma \geq 5$. This occurs because the EC decreases as the quality of the pathloss estimation decreases. This trend shows a strong impact of the proposed mode selection on the EC of the truncated HARQ-enabled D2D communication. Additionally, we observe that the impact of the quality of the pathloss estimation is significantly higher when strict QoS constraints are imposed and when the $n_1$ queue model is used. We also observe that although the $n_1$ model provides better EC, the impact of the quality of pathloss estimation is higher on the $n_1$ model compared to the $n_2$ model.

\begin{figure}[ht]
\begin{center}
	\includegraphics[width=3in]{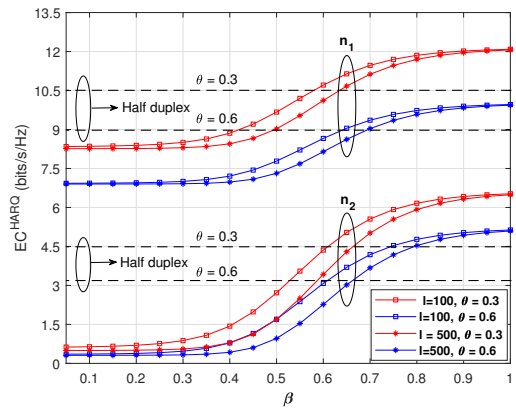}
\caption{Impact of half-duplex and full-duplex relaying on the EC of HARQ-enabled D2D communication: EC of truncated HARQ-enabled D2D vs the quality of the SI cancellation techniques.}
\label{EC_vs_beta}
\end{center}
\end{figure}
Last but not the least, we investigate the impact of half-duplex and full-duplex relaying (in mC-D2D and MC-D2D modes) on the EC of the truncated HARQ-enabled D2D communication, as shown in Fig. \ref{EC_vs_beta}. We observe that the EC increases with an increase in the quality of SI cancellation techniques ($\beta$). When $\beta$ approaches 1, it means perfect SI cancellation at the relay node (mC-BS and MC-BS), and consequently, the EC of full-duplex becomes greater than the EC of half-duplex. It is because D2D communication in half-duplex mode consumes two time-slots, and therefore, a factor of 1/2 is multiplied with the half-duplex channel capacity. On the other hand, D2D communication in full-duplex mode utilizes only one time-slot, and that is why it can achieve double throughput (theoretically) with perfect SI cancellation.  Moreover, one can also see the impact of the QoS exponent ($\theta$) and the length of the finite blocklength ($l$) on the EC of full-duplex truncated HARQ-enabled D2D communication. The EC is inversely proportional to $\theta$ and $l$; it decreases with an increase in $\theta$ and $l$ and vice-versa.
\section{Conclusion and Future Directions}
In this work, we have investigated the effects of using the HARQ protocol on the EC of buffer-aided D2D communication in multi-tier cellular networks. We have also performed the ternary hypothesis testing-based mode selection for D2D in two-tier cellular networks and have analyzed its impact on the EC of HARQ-enabled D2D communication. We have considered two different queue models at the transmitting device. In case of an outage, the transmitting device in the second model discards the packet. Whereas, in the first model, the transmitting device reduces the packet's priority rather than discarding it. We have also extended our analysis to both overlay and underlay D2D settings. Additionally, we have proposed a special case of truncated HARQ for D2D communication in which the transmitting device transmits in underlay settings in the first transmission attempt. If the receiver does not successfully decode the packet, it retransmits the packet in overlay settings in the second transmission attempt. Through simulation results, we have observed that almost three-fold enhanced EC can be achieved by using our proposed truncated HARQ protocol than by not using any retransmission protocol for D2D communication. Moreover, the first queue model provides better EC compared to the second queue model but at the expense of extra bandwidth.

Future work will study the impact of different HARQ variants on the EC of D2D communication. Moreover, this analysis can also be extended to scenarios when multiple D2D pairs are present in the network. In that case, it will be quite intriguing to investigate the impact of network and channel coding on the HARQ retransmission schemes.

\appendices
\section{pathloss Estimation}\label{pathlossprop}
The pathloss estimation has three phases, explained as follows.
\begin{itemize}
    \item Transmission Phase:
    In this phase, $D_T$ transmits $m$ number of symbols on all the candidate communication links ($D_T \to D_R$, $D_T\to BS_{_{mC}}$, and $D_T\to BS_{_{MC}}$) using fixed transmission power $P_T$. The signal received at the respective receiver ($D_R$, $BS_{_{mC}}$, and $BS_{_{MC}}$) can be calculated as follows:
    \begin{equation}
        \begin{split}
            y_{_{D_R}} &= \sqrt{P_T}\: L_d\:Z_d \:x + n_{_{d}}\\
            y_{_{mC}} &= \sqrt{P_T}\: L_{_{mC}}\:Z^{^{mC}}_{ul} \:x + n_{_{mBS}}\\
            y_{_{MC}} &= \sqrt{P_T}\: L_{_{MC}}\:Z^{^{MC}}_{ul} \:x + n_{_{MBS}},
        \end{split}
    \end{equation}

where $L_d (Z_d)$, $L_{_{mC}} (Z_{_{mC}})$, and $L_{_{MC}} (Z_{_{MC}})$ are the pathlosses (channel coefficients) between $D_T \to D_R$, $D_T \to BS_{_{mC}}$, and $D_T \to BS_{_{MC}}$, respectively, as shown in Fig. \ref{preposition1}. $x$ is the transmitted signal and $n_{_{d}}$, $n_{_{mBS}}$, $n_{_{MBS}}$ represent the noise of the respective channel. The noise of each channel follows the zero-mean complex Gaussian distribution, therefore, $n_{_{d}} \sim \mathcal{CN}(0,\sigma_{_{d}}^2)$, $n_{_{mBS}} \sim \mathcal{CN}(0,\sigma_{_{mBS}}^2)$, and $n_{_{MBS}} \sim \mathcal{CN}(0,\sigma_{_{MBS}}^2)$. We consider that the wireless channels of all the three links follow complex Gaussian distribution with zero mean and unity variance \big($Z_d \sim \mathcal{CN}(0,1)$, $Z_{_{mC}}^{_{ul}} \sim \mathcal{CN}(0,1)$, and $Z_{_{MC}}^{^{ul}} \sim \mathcal{CN}(0,1)$\big). Therefore, the received signal at all the receiver also follows the complex Gaussian distribution; $y_{_{D_R}} \sim \mathcal{CN}(0,\sigma_{_{D_R}})$, $y_{_{mC}} \sim \mathcal{CN}(0,\sigma_{_{mC}})$, and $y_{_{MC}} \sim \mathcal{CN}(0,\sigma_{_{MC}})$. To find variance of the received signal, we assume $x\in C$ and $|x|=1$, then $\sigma_{_{D_R}} = P_T \:L_{_{d}}^2 + \sigma_{_{d}}$, $\sigma_{_{mC}} = P_T \:L_{_{mC}}^2 + \sigma_{_{mBS}}$, and $\sigma_{_{MC}} = P_T \:L_{_{MC}}^2 + \sigma_{_{MBS}}$
\begin{figure}[ht]
\begin{center}
	\includegraphics[width=2.5in]{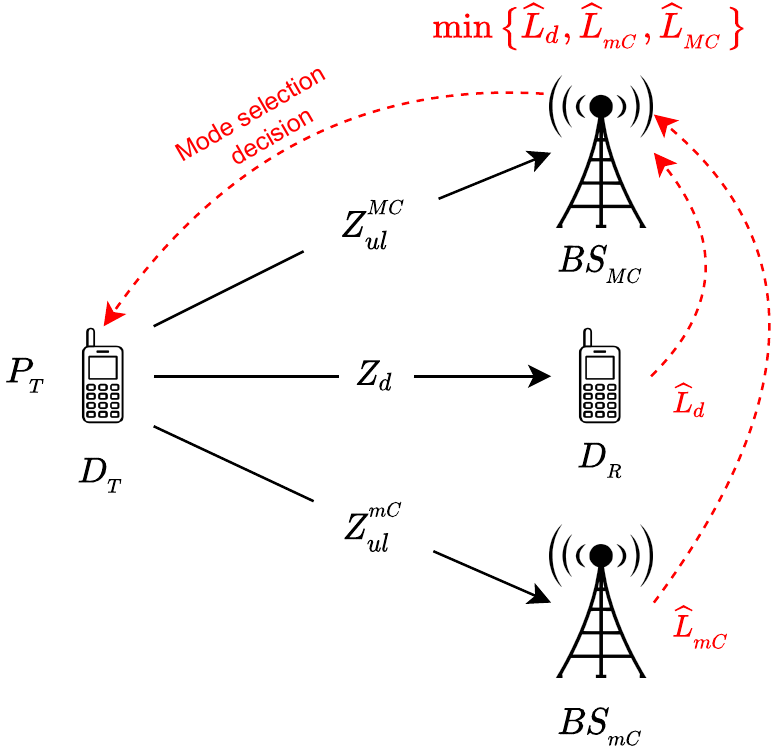}
\caption{Pathloss estimation by transmission; solid black arrows represent uplink data signaling, red dotted arrows represent uplink and downlink control signaling.}
\label{preposition1}
\end{center}
\end{figure}
\item Pathloss Estimation Phase: In this phase, every receiver estimates the pathloss of the respective communication link and then conveys it to $BS_{_{MC}}$ on the uplink control channel, which then performs the mode selection mechanism. The noisy measurement of pathloss at $D_R$, $BS_{_{mC}}$, and $BS_{_{MC}}$ can be calculated as follows:
\begin{equation}
\begin{split}
    \widehat{L}_d = \frac{\widehat{P}_{_{R,D_R}}}{P_{_{T}}}, \ &\text{where} \ \widehat{P}_{_{R,D_R}} = \frac{\sum_{i=1}^{m}|y_{_{D_R}}(i)|^2}{m}\\
    \widehat{L}_{_{mC}} = \frac{\widehat{P}_{_{R,mC}}}{P_{_{T}}}, \ &\text{where} \ \widehat{P}_{_{R,mC}} = \frac{\sum_{i=1}^{m}|y_{_{mC}}(i)|^2}{m}\\
    \widehat{L}_{_{MC}} = \frac{\widehat{P}_{_{R,MC}}}{P_{_{T}}}, \ &\text{where} \ \widehat{P}_{_{R,MC}} = \frac{\sum_{i=1}^{m}|y_{_{MC}}(i)|^2}{m}.
\end{split}
\end{equation}
$\widehat{P}_{_{R,D_R}}$, $\widehat{P}_{_{R,mC}}$, and $\widehat{P}_{_{R,MC}}$ represent the estimated values of received power at $D_R$, $BS_{_{mC}}$, and $BS_{_{MC}}$, respectively. We know that $|y_{_{D_R}}|$, $|y_{_{mC}}|$, and $|y_{_{MC}}|$ follow Rayleigh distributions and $|y_{_{D_R}}|^2$, $|y_{_{mC}}|^2$, and $|y_{_{MC}}|^2$ follow exponential distributions. Therefore, by invoking the Central Limit Theorem and for large $m$, $\widehat{P}_{_{R,D_R}}$, $\widehat{P}_{_{R,mC}}$, and $\widehat{P}_{_{R,MC}}$ follow Gaussian distributions. Similarly, $\widehat{L}_d$, $\widehat{L}_{_{mC}}$, and $\widehat{L}_{_{MC}}$ also follow Gaussian distributions.


\item Mode Selection Phase: In this phase, $BS_{_{MC}}$ obtains $\widehat{L}_d$, $\widehat{L}_{_{mC}}$, and $\widehat{L}_{_{MC}}$ for all the three candidate links via the uplink control channel. Then, it computes $\min \big\{ \widehat{L}_d, \widehat{L}_{_{mC}}, \widehat{L}_{_{MC}} \big\}$ and announces the active link via the downlink control channel to $D_T$.
\end{itemize}

\section{Proof of Proposition 1}\label{prop1}
The mC-D2D link is a two-hop wireless link consisting of an uplink and a downlink channel. Therefore, the end-to-end channel capacity of the mC-D2D link is capped by the minimum of the uplink and the downlink channel capacities \cite{farhadi2009ergodic}. It can be written as,
\begin{equation}\label{appe1}
{C^{o}_{_{mC}}}(k) = \min \{C^{^{mC}}_{_{ul}}(k),C^{^{mC}}_{_{dl}}(k)\}
\end{equation}
where $C^{^{mC}}_{_{ul}}(k)$ and $C^{^{mC}}_{_{dl}}(k)$ represent the instantaneous channel capacities of the uplink and the downlink channels of the mC-D2D mode, respectively. We assume that $BS_{_{mC}}$ operates in full-duplex mode. To cancel the self-interference caused by the simultaneous transmission and reception at $BS_{_{mC}}$, the BS utilizes the digital and analog SI cancellation techniques (see Section II-B). However, we note that in practical full-duplex systems, it is almost impossible to perfectly cancel out the effects of SI. Therefore, we incorporate residual SI as a factor of noise at the BS. Due to this, the instantaneous channel capacity of the uplink becomes,
\begin{equation}\label{appe2}
C^{^{mC}}_{_{ul}}(k) = B\log_2 \bigg(1+\frac{\bar{P} Z^{^{mC}}_{_{ul}}(k)}{L^{^{mC}}_{_{ul}}(k) N_0 + \alpha\bar{P}_{_{mC}}^{\beta}}\bigg).
\end{equation}
Where $\bar{P}_{_{mC}}$ represents the average transmit power of $BS_{_{mC}}$. $Z^{^{mC}}_{_{ul}}(k)$ and $L^{^{mC}}_{_{ul}}(k)$ represent the channel coefficients and pathloss between $D_T$ and $BS_{_{mC}}$. $\alpha\bar{P}_{_{mC}}^{\beta}$ represent residual SI, where $\alpha$ and $\beta (0\leq \beta \leq 1)$ are the constants that reflect the quality of the SI cancellation techniques employed at $BS_{_{mC}}$. By simplifying the denumerator of \eqref{appe2}, we can find the SI-to-noise-ratio for full-duplex relaying at $BS_{_{mC}}$, which is $\bar{\alpha}\bar{P}_{_{mC}}^{\beta}$, where $\bar{\alpha}=\alpha/L^{^{mC}}_{_{ul}}(k)N_0$. Next, to find the instantaneous channel capacity of the downlink of mC-D2D mode, we assume that the receiver node operates in half-duplex mode; thus, it does not experience SI. Therefore, the instantaneous channel capacity of the downlink becomes,
\begin{equation}\label{appe3}
    C^{^{mC}}_{_{dl}}(k)= B\log_2 \bigg( 1+\frac{\bar{P}_{_{mC}} Z^{^{mC}}_{_{dl}}(k)}{L^{^{mC}}_{_{dl}}(k) N_0} \bigg).
\end{equation}
Where $Z^{^{mC}}_{_{dl}}(k)$ and $L^{^{mC}}_{_{dl}}(k)$ are the channel coefficients and the pathloss between $BS_{_{mC}}$ and $D_R$, respectively. Now, by substituting \eqref{appe2} and \eqref{appe3} in \eqref{appe1}, and after some simplification steps, the end-to-end instantaneous channel capacity of mC-D2D link becomes,
\begin{equation}\label{appe4}
C^{o}_{_{mC}}(k) = \min\bigg\{ B \log_2 \big(1+ \gamma^{^{mC}}_{_{ul}}(k)\big), B \log_2 \big(1+\gamma^{^{mC}}_{_{dl}}(k)\big) \bigg\}.
\end{equation}
Where $\gamma^{^{mC}}_{_{ul}}(k) = \bar{P} Z^{^{mC}}_{_{ul}}(k)\big/1+\bar{\alpha}\bar{P}_{_{mC}}^{\beta}$ and $\gamma^{^{mC}}_{_{ul}}(k) = \bar{P}_{_{mC}} Z^{^{mC}}_{_{dl}}(k)\big/L^{^{mC}}_{_{dl}}(k) N_0$ are the SNRs of the uplink and the downlink channels, respectively. Since we are using Shannon channel capacity where the only variable that affects the channel capacity is the SNR of the transmission channel, the net-SNR of the mC-D2D link will be the minimum of uplink and downlink channels SNR. Due to this, \eqref{appe4} becomes,
\begin{equation}
C^{o}_{_{mC}}(k) = B \log_2 \big(1+\gamma_{_{mC}}(k)\big).
\end{equation}
Where $\gamma_{_{mC}}(k) = \min \big\{ \gamma^{^{mC}}_{_{ul}}(k), \gamma^{^{mC}}_{_{dl}}(k)\big\}$ is the net-SNR of mC-D2D link.

\section{Proof of Lemma 1}\label{prop2}
The block-companion matrix for queue model $n_1$ can be derived from \eqref{B_special}, which becomes,
\begin{equation}\label{le1}
\mathbf{B}_{n_1}=
\begin{bmatrix}
  b_{1,n_1} & b_{2,n_1}\\
 1 & 0
\end{bmatrix}.
\end{equation}
By solving \eqref{le1}, the largest positive root comes out to be,
\begin{equation}\label{le2}
  \lambda_{n_1}+ = \frac{1}{2}\bigg(b_{1,n_1}+\sqrt{(b_{1,n_1})^2+4(b_{2,n_1})}\bigg).
\end{equation}
To solve \eqref{le2}, we have to find $b_{1,n_1}$ and $b_{2,n_1}$. From \eqref{eq:b_k_nj}, $b_{1,n_1}$ becomes,
\begin{equation}\label{le3}
    b_{1,n_1} = \mathbf{q}_{_{1}}\mathbf{\Phi}(-\theta)\mathbf{p}_{u,i}^{\intercal}.
\end{equation}
Note that, in our proposed system, the first transmit attempt uses underlay settings. Therefore, to find $\mathbf{q}_{_{1}}$, one has to use $\Gamma_d(k)$, $\Gamma_{_{mC}}(k)$, and $\Gamma_{_{MC}}(k)$ in (\ref{decoding_error_d}) to find $\E_z[\zeta^{^{d}}_{u,1}]$, $\E_z[\zeta^{^{mC}}_{u,1}]$, and $\E_z[\zeta^{^{MC}}_{u,1}]$, respectively. Due to this fact, $\mathbf{q}_{_{1}}$ becomes $ \big[1-\E_z[\zeta^{^{d}}_{u,1}], 1, 1-\E_z[\zeta^{^{mC}}_{u,1}], 1, 1-\E_z[\zeta^{^{MC}}_{u,1}], 1\big]$. Now, by substituting $\mathbf{q}_{_{1}}$, $\mathbf{\Phi}(-\theta) = \text{diag}[e^{-lr\theta}, 1, e^{-lr\theta}, 1, e^{-lr\theta}, 1]$, and $\mathbf{p}_{u,i}=[p_{u,1}, p_{u,2}, p_{u,3}, p_{u,4}, p_{u,5}, p_{u,6}]$ in \eqref{le3}, and after some simplification steps, $b_{1,n_1}$ becomes,
\begin{equation}\label{b_1_n1}
\begin{split}
  b_{1,n_1} =& (1-\E_z[\zeta^{^{d}}_{u,1}])e^{-lr\theta}p_{u,1} + p_{u,2}+(1-\E_z[\zeta^{^{mC}}_{u,1}])e^{-lr\theta}p_{u,3}\\
  &+ p_{u,4}+(1-\E_z[\zeta^{^{MC}}_{u,1}])e^{-lr\theta}p_{u,5}+p_{u,6}\\
  &=e^{-lr\theta}\bigg[p_{u,1}\big(\alpha_{_{d}}\big)+p_{u,3}\big(\alpha_{_{mC}}\big)+p_{u,5}\big(\alpha_{_{MC}}\big)\bigg]+p^{\text{off}}_{u}.
  \end{split}
\end{equation}
Where $\alpha_d = 1-\E_z[\zeta^{^{d}}_{u,1}]$; $\alpha_{_{mC}} = 1-\E_z[\zeta^{^{mC}}_{u,1}]$; $\alpha_{_{MC}} = 1-\E_z[\zeta^{^{MC}}_{u,1}]$; and $p^{\text{off}}_{u} = p_{u,2}+p_{u,4}+p_{u,6}$, which is the sum of probabilities in OFF states for the underlay scenario.

Similarly, from \eqref{eq:b_k_nj}, $b_{2,n_1}$ becomes,
\begin{equation}\label{le4}
    b_{2,n_1} = \mathbf{q}_{_{3}}\mathbf{\Phi}(-\theta)\mathbf{p}_{o,i}^{\intercal}+\varepsilon_{_{ac}}.
\end{equation}
Note that, for the second transmit attempt, the transmit D2D node uses overlay settings for packet transmission. Therefore, to find $\mathbf{q}_{_{3}}$, $\varepsilon_d$, $\varepsilon_{_{mC}}$, and $\varepsilon_{_{MC}}$, one should use $\gamma_d(k)$, $\gamma_{_{mC}}(k)$, and $\gamma_{_{MC}}(k)$ in (\ref{decoding_error_d}). Due to this fact, $\mathbf{q}_{_{3}}$ becomes $\big[\E_z[\zeta^{^{d}}_{o,1}]-\varepsilon_d, 1, \E_z[\zeta^{^{mC}}_{o,1}]-\varepsilon_{_{mC}}, 1, \E_z[\zeta^{^{MC}}_{o,1}]-\varepsilon_{_{MC}}, 1\big]$. Now, by substituting $\mathbf{q}_{_{3}}$, $\mathbf{\Phi}(-\theta)$, and $\mathbf{p}_{o,i}=[p_{o,1}, p_{o,2}, p_{o,3}, p_{o,4}, p_{o,5}, p_{o,6}]$ in \eqref{le4}, and after some simplification steps, $b_{2,n_1}$ becomes,
\begin{equation}\label{b_2_n1}
  \begin{split}
     b_{2,n_1} =&(\E_z[\zeta^{^{d}}_{o,1}]- \varepsilon_d)e^{-lr\theta}p_{o,1} + p_{o,2}\\
     &+ (\E_z[\zeta^{^{mC}}_{o,1}]- \varepsilon_{_{mC}})e^{-lr\theta}p_{o,3} + p_{o,4}\\
      & + (\E_z[\zeta^{^{MC}}_{o,1}]- \varepsilon_{_{MC}})e^{-lr\theta}p_{o,5} + p_{o,6} + \varepsilon_{_{ac}}\\
       &= e^{-lr\theta}\bigg[p_{o,1}\big(\beta_d\big)+p_{o,3}\big(\beta_{_{mC}}\big)+p_{o,5}\big(\beta_{_{MC}}\big)\bigg]+ p^{\text{off}}_{o} + \varepsilon_{_{ac}}.
  \end{split}
\end{equation}
Where $\beta_d = \E_z[\zeta^{^{d}}_{o,1}]-\varepsilon_d$; $\beta_{_{mC}} = \E_z[\zeta^{^{mC}}_{o,1}]-\varepsilon_{_{mC}}$; $\beta_{_{MC}} = \E_z[\zeta^{^{MC}}_{o,1}]-\varepsilon_{_{MC}}$; and $p^{\text{off}}_{o} = p_{o,2}+p_{o,4}+p_{o,6}$, which is the sum of probabilities in OFF states for the overlay scenario. One can find $\varepsilon_{_{ac}}$ by calculating $\varepsilon_d$, $\varepsilon_{_{mC}}$, and $\varepsilon_{_{MC}}$ by substituting $m=2$ into (\ref{decoding_error_d}a), (\ref{decoding_error_d}b), and (\ref{decoding_error_d}c), respectively.

Now, to find $\lambda_{n_1}+$, we substitute results from (\ref{b_1_n1}) and (\ref{b_2_n1}) into (\ref{le2}). After some simplification steps, the final expression for $\lambda_{n_1}+$ becomes,
\begin{equation}\label{lambda_n1_final}
\begin{split}
  \lambda_{n_1}+ &= \frac{1}{2}\biggl(\big(e^{-lr\theta}\big[\varphi\big]+p^{\text{off}}_{u}\big)+\\
  &\sqrt{\big(e^{-lr\theta}\big[\varphi\big]+p^{\text{off}}_{u}\big)^2+4\big(e^{-lr\theta}\big[\vartheta\big]+p^{\text{off}}_{o}+\varepsilon_{_{ac}}\big)}\biggl),
\end{split}
\end{equation}
where $\varphi = p_{u,1}\big(\alpha_d\big)+p_{u,3}\big(\alpha_{_{mC}}\big)+p_{u,5}\big(\alpha_{_{MC}}\big)$ and $\vartheta = p_{o,1}\big(\beta_d\big)+p_{o,3}\big(\beta_{_{mC}}\big)+p_{o,5}\big(\beta_{_{MC}}\big)$.

\section{Proof of Lemma 2}\label{prop3}
The block-companion matrix for queue model $n_2$ can be derived from \eqref{B_special}, which becomes,
\begin{equation}\label{le21}
\mathbf{B}_{n_2}=
\begin{bmatrix}
  b_{1,n_2} & b_{2,n_2}\\
 1 & 0
\end{bmatrix}.
\end{equation}
We note that the first transmit attempt in both of the queue models uses underlay settings. Moreover, both queue models respond the same when they receive acknowledgment (either positive or negative) of the first transmit attempt, as shown in Fig. \ref{flow_chart}. Therefore, $b_{1,n_2} = b_{1,n_1}$. The expression for the second transmit attempt $b_{2,n_2}$ can be derived from \eqref{eq:b_k_nj}, which becomes
\begin{equation}\label{le222}
    b_{2,n_2} = \mathbf{q}_{_{4}}\mathbf{\Phi}(-\theta)\mathbf{p}_{o,i}^{\intercal}.
\end{equation}
Similar to $n_1$ queue model, the second transmit attempt in $n_2$ model also uses overlay settings for packet transmission. Therefore, to find $\mathbf{q}_{_{4}}$, one has to use $\gamma_d(k)$, $\gamma_{_{mC}}(k)$, and $\gamma_{_{MC}}(k)$ in (\ref{decoding_error_d}a), (\ref{decoding_error_d}b), and (\ref{decoding_error_d}c), respectively. Due to this fact, $\mathbf{q}_{_{4}}$ becomes $\big[\E_z[\zeta^{^{d}}_{o,1}], 1, \E_z[\zeta^{^{mC}}_{o,1}], 1, \E_z[\zeta^{^{MC}}_{o,1}], 1\big]$. Now, by substituting $\mathbf{q}_{_{4}}$,  $\mathbf{\Phi}(-\theta)$, and $\mathbf{p}_{o,i}$ in \eqref{le222}, and after some simplification steps, $b_{2,n_2}$ becomes,
\begin{equation}\label{b_2_n2}
\begin{split}
  b_{2,n_2}&=e^{-lr\theta}\E_z[\zeta^{^{d}}_{o,1}]p_{o,1} + p_{o,2} + e^{-lr\theta}\E_z[\zeta^{^{mC}}_{o,1}]p_{o,3}+p_{o,4}\\
  &+ e^{-lr\theta}\E_z[\zeta^{^{MC}}_{o,1}]p_{o,5}+p_{o,6}\\
  &=e^{-lr\theta}\big(p_{o,1}\E_z[\zeta^{^{d}}_{o,1}]+p_{o,3}\E_z[\zeta^{^{mC}}_{o,1}]+p_{o,5}\E_z[\zeta^{^{MC}}_{o,1}]\big)+p^{\text{off}}_{o}.
  \end{split}
\end{equation}
Now, to find the largest positive root for the case of $n_2$ ($\lambda_{n_2}+$), we substitute $b_{1,n_2}$ and $b_{2,n_2}$ into (\ref{le2}), and after some simplification steps, the final expression becomes,
\begin{equation}\label{lambda_n2_final}
\begin{split}
 \lambda_{n_2}+ = \frac{1}{2}\bigg(\big(e^{-lr\theta}&\big[\varphi\big]+p^{\text{off}}_{u}\big) +\\
  &\sqrt{\big(e^{-lr\theta}\big[\varphi\big]+p^{\text{off}}_{u}\big)^2 +4e^{-lr\theta}\big[\varrho\big]+p^{\text{off}}_{o}}\bigg)
\end{split}
\end{equation}
where $\varrho = p_{o,1}\E_z[\zeta^{^{d}}_{o,1}]+p_{o,3}\E_z[\zeta^{^{mC}}_{o,1}]+p_{o,5}\E_z[\zeta^{^{MC}}_{o,1}]$.

\footnotesize{
\bibliographystyle{IEEEtran}
\bibliography{references}
}

\vfill\break

\end{document}